\documentclass[twocolumn, journal]{IEEEtran}
\usepackage{ifpdf}
\usepackage{cite}
\usepackage{bbm}
\ifCLASSINFOpdf
\usepackage{graphicx}
\usepackage{subfigure}
\usepackage{epstopdf}
\usepackage{xcolor}
\else
\fi
\usepackage{float}
\usepackage{amsmath}
\usepackage{amssymb}
\usepackage{amsthm}
\usepackage{mathtools}
\DeclareMathOperator*{\argmin}{arg\,min}
\DeclareMathOperator*{\argmax}{arg\,max}
\newtheorem{lemma}{Lemma}
\newtheorem{prop}{Proposition}
\newtheorem{thm}{Theorem}
\newtheorem{remark}{Remark}

\usepackage{algorithm}
\usepackage{algpseudocode}

\allowdisplaybreaks
\usepackage{array}
\usepackage{stfloats}
\usepackage{tikz}

\usepackage[shortlabels]{enumitem}
\usepackage{comment}
\usepackage{multirow}
\usepackage{cleveref}
\usepackage{booktabs}
\usepackage[normalem]{ulem}

\begin{document}
\title{Markov-Enforced Discrete Diffusion Model for Digital Semantic Symbol Error Correction}

%
%
%

\author{
	\IEEEauthorblockN{Yoon Huh},
    \IEEEauthorblockN{Jeongho Kang},~\IEEEmembership{Graduate Student Member,~IEEE},
	and
	\IEEEauthorblockN{Wan Choi},~\IEEEmembership{Fellow,~IEEE}
    \thanks{Received 16 March 2026; revised 13 July 2026 and 4 September 2026; accepted 9 September 2026. This work was supported by Samsung Electronics Co., Ltd (IO251211-14376-01). The associate editor coordinating the review of this article and approving it for publication was L. Zhou. (\emph{Corresponding Authors: Wan Choi})}
	\thanks{Y.~Huh, J.~Kang and W.~Choi are with the Department of Electrical and Computer Engineering, Seoul National University (SNU), and the Institute of New Media and Communications, SNU, Seoul 08826, Korea. (e-mail: \{mnihy621, ttomas11, wanchoi\}@snu.ac.kr)}
        \vspace{-3mm}

}

%

\maketitle

\begin{abstract}
Diffusion models (DMs) have achieved remarkable success across various domains owing to their strong generative and denoising capabilities. Meanwhile, semantic communication based on neural joint source–channel coding (JSCC) has emerged as a promising paradigm for robust and efficient image transmission. However, severe channel noise can still distort the transmitted semantic symbols, resulting in significant performance degradation. Applying DMs to digital semantic symbols, particularly in vector quantization (VQ)-based systems, is fundamentally challenging because the Markov assumption does not hold for the symbol transition dynamics. To address this issue, we introduce SSCDM, a semantic symbol correcting diffusion model whose discrete-time transition dynamics are constructed using solutions from continuous-time Markov chain theory. Furthermore, to promote synergy between DMs and JSCC, our DM structure embeds discrete symbols into a latent feature space using a learned VQ codebook, and a self-organizing map-based loss is incorporated during codebook learning to enhance the geometric vicinity between neighboring digital symbols, thereby promoting topology-preserving semantic representations. Experimental results show that the proposed method significantly improves image reconstruction quality and outperforms previous symbol-level denoising techniques under low signal-to-noise ratio scenarios and different datasets.
\end{abstract}


\IEEEpeerreviewmaketitle

\section{Introduction}
Diffusion models (DMs) \cite{ho2020denoising, song2020score, austin2021structured, campbell2022continuous, sun2022score} have been widely adopted across diverse fields such as image and video synthesis, molecular design, and scientific data modeling, opening new avenues for solving complex engineering problems \cite{yang2023diffusion}. Originally developed for generative modeling, DMs can be broadly categorized according to the domain of the data they operate on. Continuous DMs \cite{ho2020denoising, song2020score} operate on continuous-valued data, where the forward process progressively corrupts continuous data by injecting noise, and a reverse denoising process is learned to iteratively recover clean representations from noisy observations. In contrast, discrete DMs \cite{austin2021structured, campbell2022continuous, sun2022score} are designed for data represented by discrete states. In this case, the forward process perturbs the data through state transitions governed by transition probability matrices across successive steps, and the reverse process reconstructs the original data by reversing these state perturbations. Importantly, both continuous and discrete DMs are fundamentally built upon a Markovian forward process, which ensures theoretical tractability and enables computationally efficient training and inference.

Although initially introduced for data generation tasks, DMs provide a unified framework for progressively modeling and reversing data corruption in both continuous and discrete domains. This property naturally aligns with wireless communication systems, where channel impairments can be interpreted as the corruption applied to transmitted signals or symbols. In principle, DMs can therefore be utilized either for continuous signal denoising at the receiver or for discrete symbol correction after detection. However, despite this clear conceptual compatibility, diffusion-based approaches remain largely unexplored in conventional wireless communication systems for the reliable transmission of random bitstreams.

Meanwhile, as sixth-generation (6G) networks shift toward task-oriented paradigms \cite{shi2023task} that demand higher spectral efficiency, semantic communication \cite{luo2022semantic} has emerged as a promising solution. Unlike conventional systems that aim to faithfully transmit raw data, semantic communication focuses on delivering task-relevant information, thereby reducing redundancy. A key enabler of this paradigm is joint source-channel coding (JSCC), which integrates source and channel coding into a unified neural network trained end-to-end for robust task performance \cite{gunduz2024joint, huh2025feature, huh2026federated, huh2026differential}. While early semantic communication schemes often adopt fully analog or continuous latent representations, practical deployments increasingly favor digital semantic communication for better compatibility with existing infrastructures. In this work, we consider a vector quantization (VQ)-based digital semantic communication system for image transmission \cite{huh2025universal, kim2026extended}, where semantic features from a JSCC encoder are quantized via a learned VQ codebook and transmitted as digital symbols. Nevertheless, despite its advantages, JSCC-based semantic communication remains vulnerable in low signal-to-noise ratio (SNR) regimes, where severe channel noise distorts semantic symbols and significantly degrades task performance.

Recently, DMs have been introduced into digital semantic communication to improve task performance under adverse wireless conditions, particularly in low-SNR regimes where semantic distortion is severe \cite{tang2025cd3m, he2025residual, mo2025scdm}. Existing studies can be divided into two distinct directions: symbol error correction using discrete DMs and symbol denoising using continuous DMs. For symbol error correction, the channel denoising discrete diffusion model (CD3M) proposed in \cite{tang2025cd3m}, following the discrete DM framework in \cite{austin2021structured}, estimates per-step transition matrices by identifying the best-fitting multiplier between two consecutive symbol detection matrices. While this approach leverages the true observed detection probabilities during the symbol error correction process, it does not explicitly account for the fact that \textit{the symbol-wise transition dynamics induced by digital communication channels violate the Markov property in the discrete domain}, which constitutes a key finding of this work. Since the outstanding performance of DMs comes from the Markov property, the violation severely limits the attainable performance. Notably, this non-Markovian behavior is an inherent characteristic of digital communication systems and is not specific to semantic communication settings. 

Another line of work on symbol error correction is the discrete channel denoising diffusion model (DCDDM) proposed in \cite{he2025residual}. It adopts a discrete DM framework \cite{austin2021structured} and constructs per-step transition matrices heuristically. Specifically, each transition matrix assigns a self-transition probability corresponding to error-free transmission, while the remaining probability mass is distributed across other symbols in proportion to their pairwise symbol error probabilities. Although this design offers a simple and computationally efficient approximation, the resulting transition matrices do not accurately capture the true stochastic behavior of digital communication channels. In particular, the off-diagonal entries do not correspond to true symbol detection error probabilities, and the independently constructed per-step transition matrices fail to \textit{jointly form a coherent Markov chain}, due to violation of the Chapman–Kolmogorov equation.

Consequently, both DCDDM and CD3M rely on transition matrices that do not satisfy the Markov property assumed in the discrete DM framework \cite{austin2021structured}, creating a fundamental mismatch between the modeled diffusion dynamics and the true channel-induced symbol transitions. The strong empirical performance and computational efficiency of DMs fundamentally arise from the Markovian structure of the forward process, which enables tractable training and inference. Retaining these advantages therefore requires strict adherence to this core theoretical premise. When the adopted transition matrices violate this rigorous foundational assumption and diverge from the actual channel dynamics, the resulting structural inconsistency severely constrains the attainable performance. The non-Markovian symbol transitions in typical digital communication therefore pose a fundamental challenge to applying discrete DMs in digital semantic communication, which serves as the central motivation for the proposed approach.

In contrast, continuous DM-based symbol denoising is exemplified by the score-based channel denoising model (SCDM) proposed in \cite{mo2025scdm}, built on a continuous DM framework \cite{song2020score}. It directly denoises noisy digital symbols, such as QAM symbols corrupted by Gaussian noise, prior to receiver-side symbol detection. However, SCDM, like the discrete DM-based approaches, operates solely on symbol representations, e.g., coordinates or indices, without \textit{explicitly accounting for the semantic relevance among encoded features}. Consequently, the diffusion module primarily serves as a noise suppressor rather than a semantics-aware component. This inherently constrains performance and hinders effective joint design with semantic communication, leaving substantial untapped potential for improvement. These limitations motivate the development of semantics-informed diffusion architectures and training strategies \textit{explicitly designed in coordination with the JSCC encoder–decoder and the learned VQ codebook}, enabling more effective learning, improved task performance, and tighter integration with digital semantic communication.

To address the above challenges, we propose a novel semantic symbol correcting diffusion model (SSCDM) tailored for VQ-based digital semantic communication. Within a discrete DM framework \cite{austin2021structured}, the core design objective of SSCDM is to improve Markov consistency across diffusion steps while preserving the channel-induced symbol transition probabilities. To this end, we leverage the shared-eigenbasis solution structure of continuous-time Markov chain (CTMC) theory \cite{campbell2022continuous} as a theoretical foundation for constructing Markov-consistent transition dynamics. We enforce this particular CTMC solution structure across all discrete-time transition matrices and jointly optimize them by minimizing the discrepancy with the true symbol transition matrices. Consequently, the resulting diffusion dynamics provide an approximate Markov projection of the channel-induced symbol transition process, enabling a more reliable reverse denoising process and significantly improving semantic reconstruction fidelity.

Furthermore, unlike prior works, SSCDM goes beyond being a simple plug-in symbol error corrector in semantic communication systems. It is explicitly designed to synergize with the JSCC framework by accounting for the latent feature space structured by the learned VQ codebook in two key aspects: latent-aware DM input embedding and a Gray-mapped codebook formulation. Specifically, SSCDM does not directly feed symbol index vectors into the DM. Rather, it first maps each index to its corresponding codeword via the learned VQ codebook, thereby forming a codeword tensor representation. The DM then operates on this semantic feature tensor, enabling it to interpret the detected symbols in a manner consistent with how the JSCC model represents semantic information. By this embedding process, SSCDM enables semantically informed and coherent diffusion-based symbol correction while preserving the discrete nature of the transmitted symbols.

As another point of synergy, we regulate the VQ codebook training to preserve the topology of constellation symbols. A VQ codebook trained without explicit inter-codeword regularization does not necessarily preserve geometric continuity or semantic consistency among neighboring symbols, unlike Gray mapping in conventional digital modulation. As a result, imperfect symbol error correction, where a misdetected symbol is moved closer to the true constellation point but does not exactly reach it, does not reliably lead to improved reconstruction performance. Although \cite{zhou2024moc} proposed a heuristic post-training approach that greedily reorders learned codewords based on Euclidean distance, such a method lacks theoretical grounding and can result in suboptimal or unstable performance. To mitigate this issue, we incorporate a self-organizing map (SOM)–based loss \cite{kohonen2002self, fortuin2018som}, which encourages neighboring constellation symbols to be associated with semantically similar codewords. This topology-preserving constraint enforces geometric and semantic consistency within the codebook, thereby further strengthening the coupling between the VQ codebook, the diffusion-based symbol correction, and the JSCC decoder. As a result, SSCDM achieves more stable symbol correction and yields substantial performance gains, particularly in low-SNR regimes.

\textbf{Contributions.}\; This paper contributes as follows:
\begin{itemize}
    \item \textbf{Transition dynamics analysis:} We analyze the inherent challenge of modeling symbol-wise transition dynamics in VQ-based digital semantic communication, highlighting the violation of the Markov property in the discrete symbol domain under digital communication channels.
    
    \item \textbf{Markov-enforced digital symbol correction method:} We propose SSCDM, a novel symbol correction framework that enforces Markov consistency on transition matrices by projecting channel-induced symbol transitions onto a CTMC-inspired transition family, enabling robust symbol correction under channel noise.
    
    \item \textbf{Neural architecture leveraging learned VQ codebook:} We design a neural JSCC architecture that utilizes a learned VQ codebook to embed discrete symbols into a continuous latent space, allowing effective diffusion-based correction with improved semantic consistency.
    
    \item \textbf{Topology-preserving VQ regularization:} To enhance the geometric consistency of the learned VQ codebook and ensure topology-preserving symbol mapping, we integrate a SOM-based regularization loss into the training, which effectively improves semantic continuity and correction robustness.
    
    \item \textbf{Comprehensive experiments:} We conduct extensive experiments across various channel conditions and datasets, demonstrating the superiority of our SSCDM in terms of semantic fidelity and robustness over existing approaches.
\end{itemize}

\textbf{Notations.}\; Vectors and matrices are expressed in lower case and upper case bold, respectively. $\|\cdot\|_p$ denotes the $l_p$ norm of a vector. $\mathbb{R}$ represents the real number set. $[a{:}b]$ is a integer set of $\{a, a + 1, \dots, b\}$. $[\cdot]^\mathsf{T}$ denotes the transpose of a matrix or vector. $[\mathbf{A}]_{i, j}$ and $[\mathbf{a}]_{i}$ denote the $(i, j)$-th entry of $\mathbf{A}$ and the $i$-th entry of $\mathbf{a}$, respectively. $[\mathbf{A}]_{i, :}$ and $[\mathbf{A}]_{:, j}$ denote the $i$-th row and the $j$-th column of $\mathbf{A}$, respectively. $\mathcal{CN}(\boldsymbol{\mu}, \boldsymbol{\Sigma})$ is a complex Gaussian distribution with mean vector $\boldsymbol{\mu}$ and covariance matrix $\boldsymbol{\Sigma}$. $\mathbf{I}_m$ and $\mathbf{1}_m$ represent the $m\times m$ identity matrix and the $m$-dimensional all-ones vector, respectively. $\mathbb{E}[\cdot]$ denotes the expected value of a given random variable.

\section{Preliminaries}\label{sec:preliminaries}

\subsection{Discrete Diffusion Model}\label{subsec:discrete diffusion model}
Analogous to the continuous DM \cite{ho2020denoising}, which removes Gaussian noise through a reverse process learned from a predefined forward noise process, the discrete DM \cite{austin2021structured} is designed to reverse the stochastic transitions applied to perturbed discrete-valued data, as illustrated in Fig. \ref{fig:system model}(c). The Markovian forward process is defined over a finite set of discrete states and must be specified in advance by the transition matrices. In this subsection, we adopt a discrete-time formulation in which transitions occur step-by-step over fixed time steps.

For $M$-categorical scalar random variables\footnote{For clarity, we describe the formulation in the one-dimensional case. Extensions to the multi-dimensional setting can be found in existing works on discrete DMs, such as \cite{austin2021structured, campbell2022continuous}.} $u_{t_k}, u_{t_{k - 1}}\in[1{:}M]$, the forward state transition matrix $\mathbf{Q}_{t_k|t_{k - 1}}\in\mathbb{R}^{M\times M}$ from time point $t_{k - 1}$ to $t_k$ is defined as $[\mathbf{Q}_{t_k|t_{k - 1}}]_{i,j} = q_{t_k|t_{k - 1}}(u_{t_k} = j|u_{t_{k - 1}} = i)$, where $0 \leq t_{k - 1} < t_k \leq 1$ for $k\in[1{:}T]$ and $T$ is a total number of times steps. Here, $u_{t_0}$ denotes the original data sampled from $q_{t_0}(u) = p_\mathsf{data}(u)$, i.e., training dataset, and $u_{t_T}$ represents the fully perturbed data following $q_{t_T}(u)$, which is a known stationary distribution approximated by $p_\mathsf{ref}(u)$ for ease of sampling.

Using the one-hot encoding of $u$, i.e., a row vector $\boldsymbol{w}\in\mathbb{R}^{1\times M}$, the transition probability can be written as $q_{t_k|t_{k - 1}}(u_{t_k}|u_{t_{k - 1}}) = \mathsf{Cat}(\boldsymbol{w}_{t_k};\mathbf{p} = \boldsymbol{w}_{t_{k - 1}}\mathbf{Q}_{t_k|t_{k - 1}})$, where $\mathsf{Cat}(\boldsymbol{w};\mathbf{p})$ denotes the categorical distribution over $\boldsymbol{w}$ with class probabilities given by the row vector $\mathbf{p}\in\mathbb{R}^{1\times M}$. This leads to the $k$-step marginal: $q_{t_k|t_0}(u_{t_k}|u_{t_0}) = \mathsf{Cat}(\boldsymbol{w}_{t_k};\mathbf{p} = \boldsymbol{w}_{t_0}\mathbf{Q}_{t_k|t_0})$, where $\mathbf{Q}_{t_k|t_0} = \prod_{i = 1}^k\mathbf{Q}_{t_i|t_{i - 1}}\in\mathbb{R}^{M\times M}$ is the accumulated transition matrix up to time $t_k$. Then, the full forward Markov process is described by
\begin{align}
    q_{t_0:t_T}(u_{t_0:t_T}) = p_\mathsf{data}(u_{t_0})\prod\nolimits_{k = 1}^{T}q_{t_k|t_{k - 1}}(u_{t_k}|u_{t_{k - 1}}).
\end{align}

On the other hand, the reverse process can be expressed as $q_{t_0:t_T}(u_{t_0:t_T}) = q_{t_T}(u_{t_T})\prod_{k = 1}^{T}q_{t_{k - 1}|t_k}(u_{t_{k - 1}}|u_{t_k})$, where $q_{t_{k - 1}|t_k} = q_{t_k|t_{k - 1}}\frac{q_{t_{k - 1}}}{q_{t_k}}$. However, in practice, this reverse kernel $q_{t_{k - 1}|t_k}$ is generally intractable, since evaluating the marginal distribution $q_{t_k}$ requires integrating over the true data distribution $p_\mathsf{data}$, which is unknown and not available in closed form, i.e., $q_{t_k}(u_{t_k}) = \int q_{t_k|t_0}(u_{t_k}|u_{t_0})p_\mathsf{data}(u_{t_0})du_{t_0}$. Thus, it is approximated by training an NN-based reverse kernel $p_{t_{k - 1}|t_k}^{\boldsymbol{\theta}}$, which is parameterized by $\boldsymbol{\theta}$. A common approach is to define the reverse kernel via the conditional distribution $q_{t_{k - 1}|t_k, t_0}$ as follows:
\begin{align}\label{eq:learned reverse kernel}
\begin{split}
    &p_{t_{k - 1}|t_k}^{\boldsymbol{\theta}}(u_{t_{k - 1}}|u_{t_k})\\
    &\quad= \sum\nolimits_{u_{t_0} = 1}^M q_{t_{k - 1}|t_k, t_0}(u_{t_{k - 1}}|u_{t_k}, u_{t_0})p_{t_0|t_k}^{\boldsymbol{\theta}}(u_{t_0}|u_{t_k}),
\end{split}
\end{align}
where $q_{t_{k - 1}|t_k, t_0} = q_{t_k|t_{k - 1}}\frac{q_{t_{k - 1}|t_0}}{q_{t_k|t_0}}$. The full reverse process is then approximated as
\begin{align}
    p_{t_0:t_T}^{\boldsymbol{\theta}}(u_{t_0:t_T}) = p_\mathsf{ref}(u_{t_T})\prod\nolimits_{k = 1}^{T}p_{t_{k - 1}|t_k}^{\boldsymbol{\theta}}(u_{t_{k - 1}}|u_{t_k}).
\end{align}

The training loss function $\mathcal{L}_\mathsf{DT}$ \cite{austin2021structured}, primarily based on the formulation in \cite{ho2020denoising}, is defined as
\begin{align}\label{eq:KL loss}
    &\nonumber\mathcal{L}_\mathsf{DT} = \mathbb{E}_{p_\mathsf{data}}\big[-\mathbb{E}_{q_{t_1|t_0}}\log{p_{t_0|t_1}^{\boldsymbol{\theta}}(u_{t_0}|u_{t_1})} + \sum\nolimits_{k = 2}^T\mathbb{E}_{q_{t_k|t_0}}\\
    &~~[\mathsf{KL}(q_{t_{k - 1}|t_k, t_0}(u_{t_{k - 1}}|u_{t_k}, u_{t_0})||p_{t_{k - 1}|t_k}^{\boldsymbol{\theta}}(u_{t_{k - 1}}|u_{t_k}))]\big],
\end{align}
which serves as an upper bound on the negative log-likelihood $\mathbb{E}_{p_\mathsf{data}}[-\log{q_{t_0}(u_{t_0})}]$. Here, the first term evaluates how accurately the model reconstructs the original data from the last single reverse step. Meanwhile, the KL term guides the model to approximate the true posterior $q_{t_{k - 1}|t_k, t_0}$ across all steps, providing a ground-truth trajectory for the reverse transition since the true reverse kernel $q_{t_{k - 1}|t_k}$ is intractable without conditioning on the original data $u_{t_0}$. To further enhance performance, a direct generative loss term is introduced:
\begin{align}\label{eq:Generative loss}
    \mathcal{L}_\mathsf{G} = \mathbb{E}_{p_\mathsf{data}}\big[\mathbb{E}_{q_{t_k|t_0}}[-\log{p_{t_0|t_k}^{\boldsymbol{\theta}}(u_{t_0}|u_{t_k})}]\big],
\end{align}
which encourages the model to directly estimate the original data from noisy states sampled at arbitrary time steps, complementing the KL term in \eqref{eq:KL loss} by explicitly considering the learned reverse kernel in \eqref{eq:learned reverse kernel}. The overall objective is given by
\begin{align}\label{eq:D3PM loss}
    \mathcal{L}_\lambda = \mathcal{L}_\mathsf{DT} + \lambda\mathcal{L}_\mathsf{G},
\end{align}
where $\lambda$ is a hyperparameter that balances the two terms.

\subsection{Continuous-Time Markov Chain}\label{subsec:continuous-time markov chain}
Under the concept of CTMC \cite{campbell2022continuous, sun2022score}, state transitions can occur at any time $t\in[0, 1]$, unlike the discrete-time Markov process in the previous subsection where transitions happen only at predefined time points $\{t_k\}_{k = 0}^T$. Given an initial distribution $q_0 = p_\mathsf{data}$, the CTMC forward process is characterized by a time-dependent rate matrix $\overrightarrow{\mathbf{R}}_t\in\mathbb{R}^{M\times M}$, which defines the instantaneous transition rates between states: $[\overrightarrow{\mathbf{R}_t}]_{\tilde{u}, u} = \overrightarrow{r_t}(u|\tilde{u}) = \lim_{\Delta t\rightarrow0}\frac{q_{t|t - \Delta t}(u|\tilde{u}) - \delta_{u, \tilde{u}}}{\Delta t}$, where $\delta_{u, \tilde{u}}$ is the Kronecker delta, equal to 1 if $u = \tilde{u}$ and 0 otherwise. Intuitively, at rate $[\overrightarrow{\mathbf{R}_t}]_{\tilde{u}, u}$, the categorical probability flows from state $\tilde{u}$ to $u$, or equivalently, the state $\tilde{u}$ is instantaneously transitions to $u$. Accordingly, $\delta_{u, \tilde{u}}$ ensures that the probability of remaining in the same state is properly accounted for, so that the sum of each row of the rate matrix is zero. Equivalently, the infinitesimal transition probability can be expressed as $q_{t|t - \Delta t}(u|\tilde{u}) = \delta_{u, \tilde{u}} + [\overrightarrow{\mathbf{R}_t}]_{\tilde{u}, u}\Delta t + o(\Delta t)$. Thus, the rate matrix satisfies the following properties.
\begin{prop}[Rate matrix properties]\label{prop:rate matrix properties}
For the CTMC forward rate matrix, it holds that
\begin{enumerate}
    \item $\overrightarrow{r_t}(u|\tilde{u}) \geq 0$ for $u \neq \tilde{u}$,
    \item $\overrightarrow{r_t}(u|u) \leq 0$,
    \item $\overrightarrow{r_t}(u|u) = -\sum\nolimits_{u \neq \tilde{u}}\overrightarrow{r_t}(u|\tilde{u})$.
\end{enumerate}
\end{prop}

Meanwhile, for $t > s$, the transition matrix $\mathbf{Q}_{t|s}$ and the rate matrix $\overrightarrow{\mathbf{R}_t}$ are related by the Kolmogorov forward equation: $\partial_t q_{t|s}(u|\tilde{u}) = \sum\nolimits_{u'}q_{t|s}(u'|\tilde{u})\overrightarrow{r_t}(u|u')$, which can be written in matrix form as $\partial_t\mathbf{Q}_{t|s} = \mathbf{Q}_{t|s}\overrightarrow{\mathbf{R}_t}$. For a particular class of CTMCs, where the rate matrices $\overrightarrow{\mathbf{R}_t}$ and $\overrightarrow{\mathbf{R}_s}$ commute for all $t$ and $s$, the solution can be expressed via eigen-decomposition and the matrix exponential $\mathsf{exp}(\cdot)$:
\begin{align}
    \mathbf{Q}_{t|0} &= \mathbf{V}\mathbf{D}(t)\mathbf{V}^{-1},\label{eq:ctmc transition matrix}\\
    \overrightarrow{\mathbf{R}_t} &= \mathbf{V}\mathbf{\Sigma}(t)\mathbf{V}^{-1},\label{eq:ctmc rate matrix}
\end{align}
where $\mathbf{D}(t) = \mathsf{exp}(\int_0^t\mathbf{\Sigma}(t')dt')\in\mathbb{R}^{M\times M}$ and $\mathbf{\Sigma}(t)\in\mathbb{R}^{M\times M}$ are diagonal matrices, and $\mathbf{V}\in\mathbb{R}^{M\times M}$ is the eigenvector matrix. Therefore, this shared-eigenbasis representation corresponds to a restricted CTMC solution structure rather than a general CTMC solution. Consequently, the sub-transition matrix from time $s$ to $t$ can be expressed as
\begin{align}\label{eq:ctmc sub-transition matrix}
    \mathbf{Q}_{t|s} = \mathbf{V}\mathsf{exp}\left(\int_s^t\mathbf{\Sigma}(t')dt'\right)\mathbf{V}^{-1}.
\end{align}

\section{System Model}\label{sec:system model}

\begin{figure}[!t]
    \centering
    \includegraphics[width=1\linewidth]{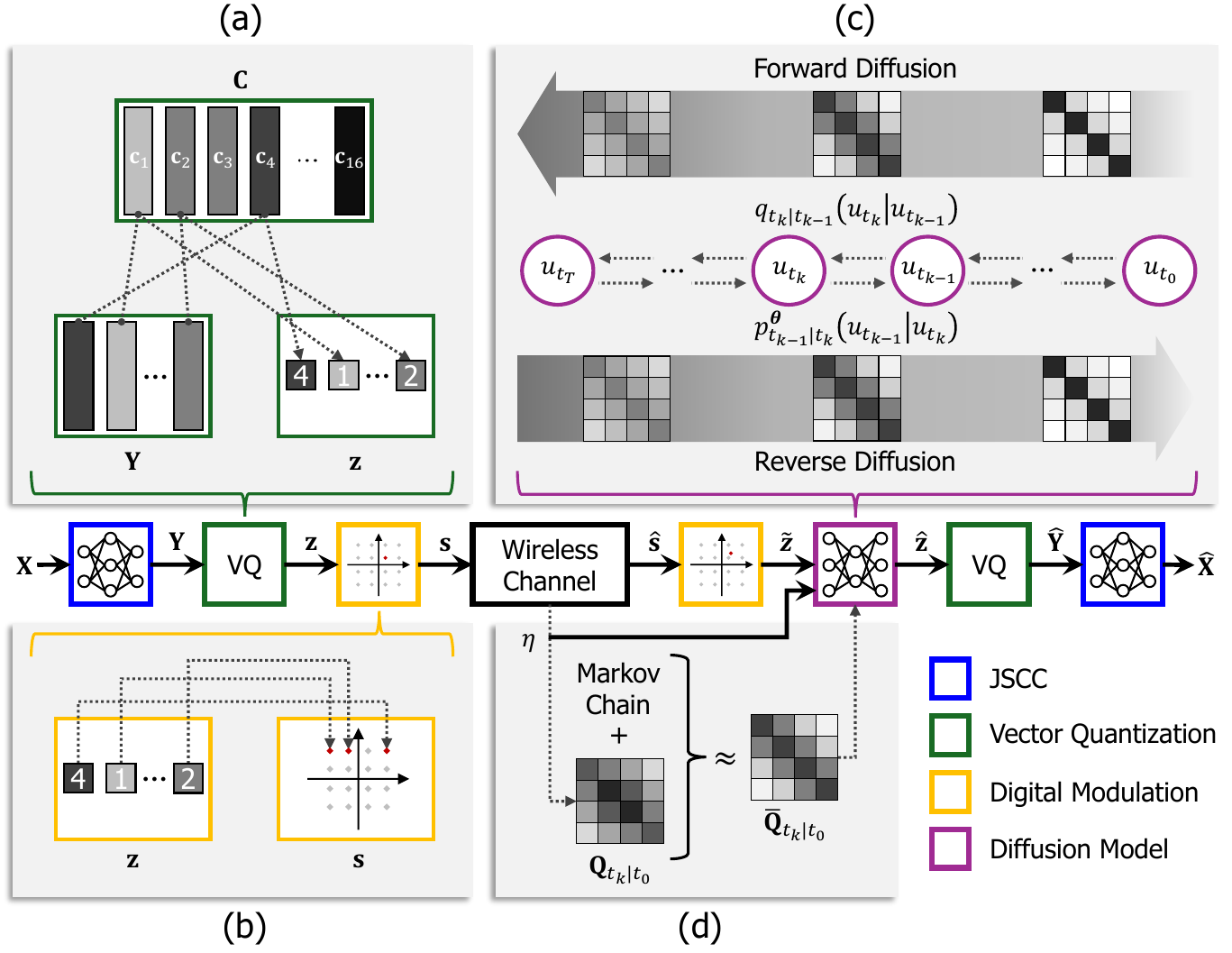}
    \caption{An illustration of (a) vector quantization ($M = 16$), (b) constellation symbol mapping ($M = 16$), (c) forward and reverse diffusion processes ($M = 4$), and (d) the design philosophy of the proposed discrete diffusion model ($M = 4$). Note that the value of $M$ differs in (a)–(d) for visualization purposes only, while it should remain consistent throughout the actual system.}
    \label{fig:system model}
\end{figure}

Consider a point-to-point digital semantic communication system for image transmission, which especially incorporates a discrete DM at the receiver to enhance robustness against detection errors on digital symbols caused by channel noise. The system employs an NN-based JSCC encoder and decoder, a VQ codebook, and a discrete DM. The overall communication procedure is as follows and summarized by Fig. \ref{fig:system model}.

At the transmitter side, the source image $\mathbf{X} \in \mathbb{R}^{C \times H \times W}$, where $C$, $H$, and $W$ represent the number of tensor channels, height, and width of the image, respectively, is passed through an encoder $f_{\boldsymbol{\psi}_\mathsf{Tx}}$ parameterized by $\boldsymbol{\psi}_\mathsf{Tx}$. This encoder converts the image into a collection of feature vectors $\mathbf{Y} = [\mathbf{y}_1, \mathbf{y}_2, \dots, \mathbf{y}_N]^\mathsf{T}\in\mathbb{R}^{N\times d}$, where $N$ is the number of feature vectors and $d$ denotes their dimensionality. Formally, this encoding step is given by $\mathbf{Y} = f_{\boldsymbol{\psi}_\mathsf{Tx}}(\mathbf{X})$.

Next, these feature vectors $\mathbf{Y}$ are quantized to produce a codeword index vector $\mathbf{z} = [z_1, z_2, \dots, z_N]^\mathsf{T}\in[1{:}M]^N$, with $M$ being the size of the codebook or digital modulation order. As shown in Fig. \ref{fig:system model}(a), the quantization uses a learned codebook $\mathbf{C} = [\mathbf{c}_1, \mathbf{c}_2, \dots, \mathbf{c}_M]^\mathsf{T}\in\mathbb{R}^{M\times d}$, selecting the nearest codeword to each feature vector $\{\mathbf{y}_i\}_{i = 1}^N$ based on Euclidean distance, i.e., $z_i = \argmin_{j} \|\mathbf{y}_i - \mathbf{c}_j\|_2^2$ for each $i\in[1{:}N]$. Each codeword $\mathbf{c}_j$ corresponds to a symbol $\mathbf{s}_j$ in the constellation set $\mathbb{S} = \{\mathbf{s}_1, \mathbf{s}_2, \dots, \mathbf{s}_M\}\in\mathbb{R}^{2\times M}$, with the symbol power normalized such that $\mathbb{E}[\|\mathbf{s}_i\|_2^2] = P$. As depicted in Fig. \ref{fig:system model}(b), the index vector $\mathbf{z}$ is then modulated into a transmitted symbol sequence $\mathbf{s} = [\mathbf{s}_{z_1}, \mathbf{s}_{z_2}, \dots, \mathbf{s}_{z_N}]\in\mathbb{S}^N$, which is sent over an additive white Gaussian noise (AWGN) channel. To summarize, the data transformation at the transmitter follows the path:
\begin{align}
    \text{(Tx)}\quad\mathbf{X} \xrightarrow{f_{\boldsymbol{\psi}_\mathsf{Tx}}} \mathbf{Y} \xrightarrow{\mathbf{C}} \mathbf{z} \xrightarrow{\mathbb{S}} \mathbf{s}.
\end{align}

At the receiver side, the transmitted symbol sequence $\mathbf{s}$ is perturbed by AWGN, i.e., $\mathbf{n} \sim \mathcal{CN}(\mathbf{0}_N, \sigma^2\mathbf{I}_N)$, yielding the received signal as $\hat{\mathbf{s}} = \mathbf{s} + \mathbf{n}$. Based on this noisy observation, a detection process estimates the transmitted symbols, producing $\bar{\mathbf{s}} = [\mathbf{s}_{\tilde{z}_1}, \mathbf{s}_{\tilde{z}_2}, \dots, \mathbf{s}_{\tilde{z}_N}]\in\mathbb{S}^N$, where $\tilde{z}_i\in[1{:}M]$ denotes the detected index for the $i$-th sybmol. The SNR for this channel is defined as $\eta = \frac{P}{\sigma^2}$, where $P$ is the average symbol power.

To compensate for the symbol detection errors caused by noise, a discrete DM, $\mathsf{DM}_{\boldsymbol{\theta}}$ parameterized by $\boldsymbol{\theta}$, is employed to refine the perturbed index sequence $\tilde{\mathbf{z}} = [\tilde{z}_1, \tilde{z}_2, \dots, \tilde{z}_N]^\mathsf{T}\in[1{:}M]^N$. Each index $\tilde{z}_i$ corresponds to the $M$-categorical random variable $u_{t_k}$ in Fig. \ref{fig:system model}(c), representing a noise-perturbed state. The model progressively corrects the corrupted state and produces a restored index vector $\hat{\mathbf{z}} = [\hat{z}_1, \hat{z}_2, \dots, \hat{z}_N]^\mathsf{T}\in[1{:}M]^N$. This reverse diffusion process is expressed as $\hat{\mathbf{z}} = \mathsf{DM}_{\boldsymbol{\theta}}(\tilde{\mathbf{z}}, \eta, \mathbf{C})$, where the SNR value $\eta$ determines the number of denoising steps, i.e., the starting time point $t_k$, according to the severity of index perturbation. Notably, the proposed diffusion model performs semantic-aware processing by leveraging the learned codebook $\mathbf{C}$, which is detailed in Section \ref{subsec:semantics-informed diffusion model architecture}. Consequently, $\mathsf{DM}_{\boldsymbol{\theta}}$ interprets the detected symbol indices $\tilde{\mathbf{z}}$ in a manner consistent with the semantic feature representation of the JSCC framework. Each restored index $\hat{z}_i$ corresponds to $u_{t_0}$ in Fig. \ref{fig:system model}(c), denoting the clean target state in the reverse diffusion trajectory.

From the restored indices $\hat{\mathbf{z}}$, the receiver reconstructs the quantized feature vectors $\hat{\mathbf{Y}} = [\mathbf{c}_{\hat{z}_1}, \mathbf{c}_{\hat{z}_2}, \dots, \mathbf{c}_{\hat{z}_N}]^\mathsf{T}\in\mathbb{R}^{N\times d}$ by referencing the same codebook $\mathbf{C}$ used at the transmitter. Finally, a decoder $g_{\boldsymbol{\psi}_\mathsf{Rx}}$, parameterized by $\boldsymbol{\psi}_\mathsf{Rx}$, transforms the recovered features into the reconstructed image $\hat{\mathbf{X}}\in\mathbb{R}^{C\times H\times W}$. This decoding step can be expressed as $\hat{\mathbf{X}} = g_{\boldsymbol{\psi}_\mathsf{Rx}}(\hat{\mathbf{Y}})$. The data flow at the receiver can be outlined as follows:
\begin{align}
    \text{(Rx)}\quad\hat{\mathbf{s}} \xrightarrow{\textsf{Detection}} \bar{\mathbf{s}} \xrightarrow{\mathbb{S}} \tilde{\mathbf{z}} \xrightarrow{\mathsf{DM}_{\boldsymbol{\theta}}} \hat{\mathbf{z}} \xrightarrow{\mathbf{C}} \hat{\mathbf{Y}} \xrightarrow{g_{\boldsymbol{\psi}_\mathsf{Rx}}} \hat{\mathbf{X}}.
\end{align}

\section{Semantic Symbol Error Correcting\\Discrete Diffusion Model}\label{sec:semantic symbol error correcting discrete diffusion model}
In this section, we propose a discrete DM to correct symbol errors in VQ-based digital semantic communication systems. We first identify a fundamental challenge in designing state transition matrices that accurately reflect channel-induced symbol corruption. To address this challenge, we propose a novel solution termed SSCDM, which leverages CTMC theory to construct Markov-consistent transition matrices for a discrete-time diffusion process. These transition matrices are obtained prior to DM training and kept fixed during DM training and sampling, with the detailed procedures provided in Section \ref{sec:training and sampling algorithms}. Finally, we present the semantically informed NN architecture used to implement SSCDM.

\subsection{Discussions on Transition Matrix Design}\label{subsec:discussions on transition matrix design}
In a discrete diffusion process for VQ-based semantic communication, the forward transition matrices $\{\mathbf{Q}_{t_k|t_0}\}_{k = 1}^T$ should reflect the true state transition probabilities induced when constellation points in $\mathbb{S}$ are perturbed by Gaussian noise and then detected as discrete symbols. Based on this principle, the following conditions must hold.
\begin{prop}[Transition matrix condition]\label{prop:transition matrix condition}
Let $\sigma_{t_k}^2$ denote the complex Gaussian noise variance added from time point $t_{k - 1}$ to $t_k$ for $k\in[1{:}T]$ with $\sigma_{t_0}^2 = 0$. Given this predefined noise variance sequence $\{\sigma_{t_k}^2\}_{k = 1}^T$, the transition matrices for symbols corrupted by Gaussian noise must satisfy
\begin{align*}
    [\mathbf{Q}_{t_k|t_0}]_{i, j} = \iint_{\mathbb{A}_j}\frac{1}{\pi\bar{\sigma}_{t_k}^2}e^{-\frac{||\mathbf{x} - \mathbf{s}_i||_2^2}{\bar{\sigma}_{t_k}^2}}d\mathbf{x},
\end{align*}
representing the standard symbol detection probability that a noisy received point generated from $\mathbf{s}_i$ is detected as the $j$-th symbol $\mathbf{s}_j$ \cite{proakis2008digital}, where $\mathbb{A}_j$ is the decision region for $\mathbf{s}_j$ and $\bar{\sigma}_{t_k}^2$ is the cumulative noise variance given by $\sum_{\ell = 1}^k\sigma_{t_\ell}^2$. For $k \in [2{:}T]$, the following must also hold:
\begin{align*}
    [\mathbf{Q}_{t_k|t_{k - 1}}]_{i, j} = \iint_{\mathbb{A}_i}\frac{a_{t_{k - 1}}(\mathbf{x})}{b_{t_{k - 1}}(i)}\iint_{\mathbb{A}_j}\frac{1}{\pi\sigma_{t_k}^2}e^{-\frac{||\mathbf{x}' - \mathbf{x}||_2^2}{\sigma_{t_k}^2}}d\mathbf{x}'d\mathbf{x},
\end{align*}
corresponding to the transition probability from a region $\mathbb{A}_i$ to a region $\mathbb{A}_j$, where $a_{t_k}(\mathbf{x})$ denotes the marginal distribution over the constellation at time $t_k$, expressed as $a_{t_k}(\mathbf{x}) = \sum_{u = 1}^M\frac{q_{t_0}(u)}{\pi\bar{\sigma}_{t_k}^2}e^{-||\mathbf{x} - \mathbf{s}_u||_2^2/\bar{\sigma}_{t_k}^2}$, and $b_{t_k}(i)$ is a normalization factor for $a_{t_k}(\mathbf{x})$ in a region $\mathbb{A}_i$, such as $b_{t_k}(i) = \iint_{\mathbb{A}_i}a_{t_k}(\mathbf{x})d\mathbf{x}$.
\end{prop}

However, for standard decision regions of typical digital modulation schemes (e.g., QAM, PAM, and PSK) under nonzero Gaussian noise, the transition matrices satisfying the above conditions exhibit non-Markovian behavior, as stated in the following theorem.

\begin{thm}\label{thm:Markov is impossible}
For typical digital modulation schemes under nonzero Gaussian noise, the transition matrices satisfying the conditions in \emph{\textbf{Proposition \ref{prop:transition matrix condition}}} do not form a Markov chain.
\begin{IEEEproof}
    See \emph{Appendix \ref{apdx:Markov is impossible}}.
\end{IEEEproof}
\end{thm}

\begin{remark}\label{rm:Markov is impossible}
Following \emph{\textbf{Proposition \ref{prop:transition matrix condition}}} with a uniformly distributed prior $q_{t_0}$, the violation of Markov property can also be numerically examined using the Frobenius norm error, i.e., $e_{t_{k_2}|t_{k_1}} = \|\mathbf{Q}_{t_{k_2}|t_0} - \mathbf{Q}_{t_{k_1}|t_0}\mathbf{Q}_{t_{k_2}|t_{k_1}}\|_\mathsf{F}$ for $t_0 < t_{k_1} < t_{k_2}$. While $\mathbf{Q}_{t_{k_1}|t_0}$ and $\mathbf{Q}_{t_{k_2}|t_0}$ are obtained analytically, $\mathbf{Q}_{t_{k_2}|t_{k_1}}$ is estimated via Monte Carlo (MC) simulation with $10^6$ samples per 16-QAM symbol. For the time pairs $(t_{k_1}, t_{k_2}) = (t_{9}, t_{20})$ and $(t_{40}, t_{65})$ under the experimental setup in \emph{Section \ref{sec:experimental results}}, the errors are $0.2412$ and $0.2603$, respectively. These values are considerably larger than the fluctuations from the MC estimation, i.e., $0.0008$ and $0.0002$, computed as the difference between two independent errors such as $|e_{t_{k_2}|t_{k_1}}^{(1)} - e_{t_{k_2}|t_{k_1}}^{(2)}|$.
\end{remark}

This result reveals a fundamental limitation in applying discrete DMs to symbol-level correction. The transition matrices induced by Gaussian perturbation and subsequent symbol detection do not satisfy the Markov consistency condition required by most discrete DMs \cite{austin2021structured, campbell2022continuous}. This discrepancy between the digital communication channel-induced transition dynamics and the Markov structure assumed by existing discrete DMs can lead to mismatches during training or inference, degrading the performance of symbol error correction in digital semantic communication systems.

\subsection{Markov-Enforced Transition Matrices}\label{subsec:markov-enforced transition matrices}
To address the non-Markovian issue identified in \textbf{Theorem \ref{thm:Markov is impossible}}, we propose SSCDM, which enforces Markov consistency on the channel-induced symbol transition matrices $\{\mathbf{Q}_{t_k|t_0}\}_{k = 1}^{T}$ in \textbf{Proposition \ref{prop:transition matrix condition}} to a Markov-consistent transition family under the predefined noise variance sequence\footnote{We adopt a sigmoid-based scheduling function \cite{jabri2022scalable}, as its smooth and non-linear noise allocation enables stable reverse diffusion under varying channel SNR conditions. Details are provided in Appendix \ref{apdx:noise scheduling method}.} $\{\sigma_{t_k}^2\}_{k = 1}^T$, as conceptualized in Fig. \ref{fig:system model}(d). While the proposed SSCDM operates within the discrete-time diffusion framework, we draw inspiration from the CTMC solution in \cite{campbell2022continuous}, reviewed in Section \ref{subsec:continuous-time markov chain}, to construct a set of Markov-consistent transition matrices. Specifically, leveraging the shared-eigenbasis structure, we formulate a CTMC-constrained fitting problem within this restricted class for the time-sampled transition matrices $\{\mathbf{\bar{Q}}_{t_k|t_0}\}_{k = 1}^{T}$ as follows.
\begin{subequations}
\begin{align}
    (\textsf{P0})\quad\quad\quad\argmin_{\mathbf{V}, \mathbf{\Sigma}(t)}&\sum_{k = 1}^T||\mathbf{Q}_{t_k|t_0} - \mathbf{\bar{Q}}_{t_k|t_0}||_\mathsf{F}^2,\label{eq:p0}\\
    \text{s.t.}\quad\quad\mathbf{\bar{Q}}_{t_k|t_0} &= \mathbf{V}\mathbf{D}(t_k)\mathbf{V}^{-1}~\text{for}~k\in[1{:}T],\label{eq:p0c1}\\
    \nonumber\mathbf{\bar{Q}}_{t_k|t_{k - 1}} &= \mathbf{V}\mathbf{D}(t_{k - 1})^{-1}\mathbf{D}(t_k)\mathbf{V}^{-1}\\
    &\quad\quad\quad\quad\quad\quad\quad\text{for}~k\in[1{:}T],\label{eq:p0c2}\\
    [\mathbf{\bar{Q}}_{t_k|t_0}]_{i, j} &\geq 0~\text{for}~i,\! j\!\in[1{:}M], k\in[1{:}T],\label{eq:p0c3}\\
    [\mathbf{\bar{Q}}_{t_k|t_{k - 1}}]_{i, j} &\geq 0~\text{for}~i,\! j\!\in[1{:}M], k\in[1{:}T],\label{eq:p0c4}\\
    \mathbf{D}(t_k) &= \mathsf{exp}\left(\int_0^{t_k}\mathbf{\Sigma}(t')dt'\right)~\text{for}~k\in[1{:}T],\label{eq:p0c5}\\
    \overrightarrow{\mathbf{R}_t} &= \mathbf{V}\mathbf{\Sigma}(t)\mathbf{V}^{-1}~\text{for}~0\leq t\leq 1,\label{eq:p0c6}\\
    [\overrightarrow{\mathbf{R}_t}]_{i, i} &\leq 0~\text{for}~i\in[1{:}M], 0\leq t\leq 1,\label{eq:p0c7}\\
    [\overrightarrow{\mathbf{R}_t}]_{i, j\neq i} &\geq 0~\text{for}~i,\! j\!\in[1{:}M], 0\leq t\leq 1,\label{eq:p0c8}\\
    [\mathbf{V}]_{:, 1} &= \mathbf{1}_M,\label{eq:p0c9}\\
    [\mathbf{\Sigma}(t)]_{1, 1} &= 0~\text{for}~0\leq t\leq 1.\label{eq:p0c10}
\end{align}
\end{subequations}
\eqref{eq:p0c1} and \eqref{eq:p0c2} impose the CTMC solution structure on the cumulative and one-step transition matrices, respectively, based on the eigen-decomposition in \eqref{eq:ctmc transition matrix} and \eqref{eq:ctmc sub-transition matrix}. \eqref{eq:p0c3} and \eqref{eq:p0c4} enforce the non-negativity of both the cumulative and one-step transition matrices. \eqref{eq:p0c5} and \eqref{eq:p0c6} specify the relationship between the transition matrix and the rate matrix in \eqref{eq:ctmc rate matrix} and \eqref{eq:ctmc sub-transition matrix}. \eqref{eq:p0c7}-\eqref{eq:p0c10} are imposed so that the rate matrix $\overrightarrow{\mathbf{R}_t}$ satisfies the properties in \textbf{Proposition \ref{prop:rate matrix properties}}. Especially, \eqref{eq:p0c9} and \eqref{eq:p0c10} enable the row-sums of $\overrightarrow{\mathbf{R}_t}$ and $\mathbf{\bar{Q}}_{t_k|t_0}$ to be $0$ and $1$, respectively: 1) $\sum_{i = 1}^M[\overrightarrow{\mathbf{R}_t}]_{j, i} = [\overrightarrow{\mathbf{R}_t}]_{j, :}\mathbf{1}_M = [\overrightarrow{\mathbf{R}_t}]_{j, :}[\mathbf{V}]_{:, 1} = [\mathbf{V}]_{j, 1}[\mathbf{\Sigma}(t)]_{1, 1} = 0$ for all $j\in[1{:}M]$, where the second equality follows from \eqref{eq:p0c9}; the third equality follows from \eqref{eq:p0c6}; the last equality follows from \eqref{eq:p0c10}; 2) $\sum_{i = 1}^M[\mathbf{\bar{Q}}_{t_k|t_0}]_{j, i} = [\mathbf{\bar{Q}}_{t_k|t_0}]_{j, :}\mathbf{1}_M = [\mathbf{\bar{Q}}_{t_k|t_0}]_{j, :}[\mathbf{V}]_{:, 1} = [\mathbf{V}]_{j, 1}[\mathbf{D}(t_k)]_{1, 1} = [\mathbf{V}]_{j, 1}\cdot 1 = 1$ for all $j\in[1{:}M]$, where the third equality follows from \eqref{eq:p0c1}; the fourth equality follows from \eqref{eq:p0c5} and \eqref{eq:p0c10}; the last equality follows from \eqref{eq:p0c9}. These imply $\overrightarrow{\mathbf{R}_t}\mathbf{1}_M = 0\cdot\mathbf{1}_M$ and $\mathbf{\bar{Q}}_{t_k|t_0}\mathbf{1}_M = 1\cdot\mathbf{1}_M$. Since eigenvectors have scaling freedom and the eigenvector corresponding to the unit eigenvalue must be parallel to $\mathbf{1}_M$, we set the first eigenvalue–eigenvector pair of $\overrightarrow{\mathbf{R}_t}$ to $(0, \mathbf{1}_M)$ without loss of generality.

However, directly solving \textsf{P0} is computationally intractable in practice. First, \textsf{P0} requires optimizing the continuous-time function $\mathbf{\Sigma}(t)$ over the entire interval $0\leq t\leq1$, which necessitates a complex functional parameterization. Second, the matrix exponential constraint in \eqref{eq:p0c5} makes the problem highly nonlinear and non-convex. Moreover, the proposed discrete DM does not explicitly use the rate matrix during training or inference, and it only requires the transition matrices at the predefined diffusion time points $\{t_k\}_{k = 1}^{T}$. Therefore, we reduce \textsf{P0} to the following finite-dimensional transition-matrix fitting problem, whose solution is directly used to define the discrete-time diffusion process.

\begin{subequations}
\begin{align}
    (\textsf{P1})\quad\;\argmin_{\mathbf{V}, \{\mathbf{D}(t_k)\}_{k = 1}^T}&\sum_{k = 1}^T||\mathbf{Q}_{t_k|t_0} - \mathbf{\bar{Q}}_{t_k|t_0}||_\mathsf{F}^2,\label{eq:p1}\\
    \nonumber\text{s.t.}\quad\quad\eqref{eq:p0c1},&~\eqref{eq:p0c2},~\eqref{eq:p0c3},~\eqref{eq:p0c4},~\eqref{eq:p0c9}\\
    [\mathbf{D}(t_k)]_{1, 1} &= 1~\text{for}~k\in[1{:}T],\label{eq:p1c6}\\
    [\mathbf{D}(t_k)]_{i, i} &\geq 0~\text{for}~i\in[2{:}M], k\in[1{:}T],\label{eq:p1c7}\\
    [\mathbf{D}(t_{k - 1})]_{i, i} &\geq [\mathbf{D}(t_k)]_{i, i}~\text{for}~i\in[2{:}M], k\in[1{:}T].\label{eq:p1c8}
\end{align}
\end{subequations}
\eqref{eq:p1c6} follows from \eqref{eq:p0c5} and \eqref{eq:p0c10}, since $[\mathbf{D}(t_k)]_{1, 1} = \mathsf{exp}(\int_0^{t_k}[\mathbf{\Sigma}(t')]_{1, 1}dt') = 1$. Together, \eqref{eq:p0c9} and \eqref{eq:p1c6} guarantee that each $\mathbf{\bar{Q}}_{t_k|t_0}$ satisfies the unit row-sum property, i.e., $\sum_{i = 1}^M[\mathbf{\bar{Q}}_{t_k|t_0}]_{j, i} = 1$ for all $j\in[1{:}M]$. \eqref{eq:p1c7} ensures the non-negativity of $[\mathbf{D}(t_k)]_{i, i}$, which is consistent with the exponential form in \eqref{eq:p0c5}. Lastly, \eqref{eq:p1c8} imposes the non-increasing property of the diagonal elements of $\mathbf{D}(t_k)$ over diffusion time points. Here, we define $\mathbf{D}(t_0) = \mathbf{I}_M$ based on the identity $\mathbf{Q}_{t_0|t_0} = \mathbf{I}_M$. Since \textsf{P1} relaxes the explicit rate matrix constraints in \textsf{P0}, \eqref{eq:p1c8} is introduced as a necessary condition for the fitted transition family to be compatible with an underlying CTMC rate matrix representation, as shown in the following lemma.

\begin{lemma}\label{lm:non-increasing diagonal}
Assume that there exist rate matrices $\overrightarrow{\mathbf{R}_t}$ corresponding to the solution of \emph{\textsf{P1}} for $0 \leq t \leq 1$. Then, for $i\in[1{:}M]$, a time sequence of each diagonal element, i.e., $\{[\mathbf{D}(t_{k})]_{i, i}\}_{k = 0}^{T}$, is non-increasing.
\begin{IEEEproof}
    See \emph{Appendix \ref{apdx:non-increasing diagonal}}.
\end{IEEEproof}
\end{lemma}

Although $\textsf{P1}$ is more tractable than $\textsf{P0}$, it remains non-convex and involves many time steps $T$, making the optimization computationally demanding. In addition, jointly estimating all transition matrices at once may lead to suboptimal approximations due to accumulated fitting errors. To address these challenges, we sample a subsequence $\{t_{k_\ell}\}_{\ell = 1}^{T'}$ including $T'$ time steps with $T' < T$, where $k_\ell\in[1{:}T]$ and $k_{T'} = T$, and solve the following reduced problem as an initial step.
\begin{align}
    (\textsf{P2})\quad\argmin_{\mathbf{V}, \{\mathbf{D}(t_{k_\ell})\}_{\ell = 1}^{T'}}&\sum_{\ell = 1}^{T'}||\mathbf{Q}_{t_{k_\ell}|t_0} - \mathbf{\bar{Q}}_{t_{k_\ell}|t_0}||_\mathsf{F}^2,\label{eq:p2}\\
    \nonumber\text{s.t.}\quad\quad\mathbf{\bar{Q}}_{t_{k_\ell}|t_0} &= \mathbf{V}\mathbf{D}(t_{k_\ell})\mathbf{V}^{-1}~\text{for}~\ell\in[1{:}T'],\\
    \nonumber\mathbf{\bar{Q}}_{t_{k_\ell}|t_{k_{\ell - 1}}} &= \mathbf{V}\mathbf{D}(t_{k_{\ell - 1}})^{-1}\mathbf{D}(t_{k_\ell})\mathbf{V}^{-1}\\
    \nonumber&\quad\quad\quad\quad\quad\quad\quad\;\text{for}~\ell\in[1{:}T'],\\
    \nonumber[\mathbf{\bar{Q}}_{t_{k_\ell}|t_0}]_{i, j} &\geq 0~\text{for}~i, j\in[1{:}M], \ell\in[1{:}T'],\\
    \nonumber[\mathbf{\bar{Q}}_{t_{k_\ell}|t_{k_{\ell - 1}}}]_{i, j} &\geq 0~\text{for}~i, j\in[1{:}M], \ell\in[1{:}T'],\\
    \nonumber[\mathbf{V}]_{:, 1} &= \mathbf{1}_M,\\
    \nonumber[\mathbf{D}(t_{k_\ell})]_{1, 1} &= 1~\text{for}~\ell\in[1{:}T'],\\
    \nonumber[\mathbf{D}(t_{k_\ell})]_{i, i} &\geq 0~\text{for}~i\in[2{:}M],\ell\in[1{:}T'],\\
    \nonumber[\mathbf{D}(t_{k_{\ell - 1}})]_{i, i} &\geq [\mathbf{D}(t_{k_\ell})]_{i, i}~\text{for}~i\in[2{:}M], \ell\in[1{:}T'].
\end{align}

\textsf{P2} estimates only the coarse-grained transition matrices $\{\mathbf{\bar{Q}}_{t_{k_\ell}|t_0}\}_{\ell = 1}^{T'}$ under $\{\sigma_{t_{k_\ell}}^2\}_{\ell = 1}^{T'}$, rather than the full set $\{\mathbf{\bar{Q}}_{t_{k}|t_0}\}_{k = 1}^{T}$. To solve this optimization, we adopt a block coordinate descent strategy using the Adam optimizer, where the common eigenvector matrix $\mathbf{V}$ and the set of diagonal matrices $\{\mathbf{D}(t_{k_\ell})\}_{\ell = 1}^{T'}$ are alternatively updated\footnote{Due to the structural complexity and non-convexity involving matrix inverses, we employ the Adam optimizer, a stochastic gradient descent--based method, within an automatic-differentiation framework. Note that Adam is used here as a general-purpose adaptive gradient optimizer for the matrix variables, not as a conventional NN training procedure. This practical approach simplifies optimization by avoiding complicated or non-trivial derivation of complex matrix gradients while achieving high approximation fidelity, as confirmed in Section \ref{subsec:transition matrix similarity}.}. Through this alternating minimization, we reduce the approximation error between the original transition matrices and their eigen-decomposed counterparts, while enforcing the stochastic properties required for valid transition matrices. The overall optimization procedure is detailed in \textbf{Algorithm \ref{alg:block coordinate descent algorithm}}\footnote{A diagonal matrix $\mathbf{D}(t_k)$ is defined by its diagonal vector $\mathsf{diag}(\mathbf{D}(t_k))\in\mathbb{R}^M$, where $\mathsf{diag}(\mathbf{A})$ returns the vector containing the diagonal entries of a matrix $\mathbf{A}\in\mathbb{R}^{M\times M}$.}.

The \textit{row-stochasticity} of $\{\mathbf{\bar{Q}}_{t_{k_\ell}|t_0}\}_{\ell = 1}^{T'}$ is ensured by fixing $[\mathbf{V}]_{:, 1} = \mathbf{1}_M$ and $[\mathbf{D}(t_{k_\ell})]_{1, 1} = 1$ for $\ell\in[1{:}T']$, as described in line \ref{line:row-stochasticity} of \textbf{Algorithm \ref{alg:block coordinate descent algorithm}}. Accordingly, they remain unchanged during the optimization process. The objective function for the optimization is defined as $\mathcal{L}_{\textsf{P2}} = \sum_{\ell = 1}^{T'}\mathcal{L}_{t_{k_\ell}}$, where
\begin{align}\label{eq:bcd loss function}
    \mathcal{L}_{t_{k_\ell}} &= ||\mathbf{Q}_{t_{k_\ell}|t_0} - \bar{\mathbf{Q}}_{t_{k_\ell}|t_0} ||_\mathsf{F}^2 + \lambda_1||\mathsf{ReLU}(-\bar{\mathbf{Q}}_{t_{k_\ell}|t_0})||_\mathsf{F}^2\\
    \nonumber&+ \lambda_2||\mathsf{ReLU}(-\bar{\mathbf{Q}}_{t_{k_\ell}|t_{k_{\ell - 1}}})||_\mathsf{F}^2 + \lambda_3||\mathsf{ReLU}(-\mathbf{D}(t_{k_\ell}))||_\mathsf{F}^2.
\end{align}
Here, $\lambda_1$, $\lambda_2$ and $\lambda_3$ are weight hyperparameters, and $\mathsf{ReLU}(\mathbf{A})$ denotes the element-wise rectified linear unit (ReLU) function defined as $\max(0, [\mathbf{A}]_{i, j})$ for all entries of a matrix $\mathbf{A}\in\mathbb{R}^{M\times M}$. The first term in \eqref{eq:bcd loss function} represents the primary approximation loss. The second and third terms penalize negative entries in the cumulative transition matrices at the subsequence time points and the one-step transition matrices between adjacent subsequence time points, respectively, to enforce non-negativity required for valid transition matrices. The last term induces the non-negativity of the diagonal matrices. After updating the diagonal entries, we explicitly \textit{enforce} this non-negativity constraint, as written in line \ref{line:enforce} of \textbf{Algorithm \ref{alg:block coordinate descent algorithm}}. In addition, to impose the \textit{non-increasing} property of the diagonal elements in \textbf{Lemma \ref{lm:non-increasing diagonal}}, we apply the pool adjacent violators algorithm (PAVA) \cite{best1990active}, a widely used non-increasing projection method, to the sequence $[\mathbf{D}(t_{k_\ell})]_{i, i}$ for each $i\in[2{:}M]$, as written in line \ref{line:non-increasing} of \textbf{Algorithm \ref{alg:block coordinate descent algorithm}}.

\begin{algorithm}[!t]
\caption{Block coordinate descent algorithm for \textsf{P2}}\label{alg:block coordinate descent algorithm}
\begin{algorithmic}[1]
\State{\textbf{Input:} Transition matrices $\{\mathbf{Q}_{t_{k_\ell}|t_0}\}_{\ell = 1}^{T'}$, weight hyperparameters $\lambda_1$, $\lambda_2$ and $\lambda_3$}
\State{Initialize $\mathbf{V}$ with $[\mathbf{V}]_{:, 1} = \mathbf{1}_M$ and $\{\mathsf{diag}(\mathbf{D}(t_{k_\ell}))\}_{\ell = 1}^{T'}$ with $[\mathbf{D}(t_{k_\ell})]_{1, 1} = 1$ and $[\mathbf{D}(t_{k_\ell})]_{i, i} \geq 0$}
\State{Freeze $[\mathbf{V}]_{:, 1}$ and $\{[\mathbf{D}(t_{k_\ell})]_{1, 1}\}_{\ell = 1}^{T'}$}\label{line:row-stochasticity}\Comment{\textit{row-stochasticity}}
\While{not converged}
    \Statex{\textit{\textbf{Step 1: Update $\mathbf{V}$}}}
    \State{Freeze $\{\mathbf{D}(t_{k_\ell})\}_{\ell = 1}^{T'}$}
    \State{$\{\mathbf{\bar{Q}}_{t_{k_\ell}|t_0}\}_{\ell = 1}^{T'} \gets \{\mathbf{V}\mathbf{D}(t_{k_\ell})\mathbf{V}^{-1}\}_{\ell = 1}^{T'}$}
    \State{$\{\mathbf{\bar{Q}}_{t_{k_\ell}|t_{k_{\ell - 1}}}\}_{\ell = 1}^{T'} \gets \{\mathbf{V}\mathbf{D}(t_{k_{\ell - 1}})^{-1}\mathbf{D}(t_{k_\ell})\mathbf{V}^{-1}\}_{\ell = 1}^{T'}$}
    \State{Compute loss $\mathcal{L}_{\textsf{P2}}$ referring \eqref{eq:bcd loss function}}
    \State{Update $\mathbf{V}$ via Adam optimizer}
    \Statex{\textit{\textbf{Step 2: Update $\{\mathbf{D}(t_{k_\ell})\}_{\ell = 1}^{T'}$}}}
    \State{Freeze $\mathbf{V}$}
    \State{$\{\mathbf{\bar{Q}}_{t_{k_\ell}|t_0}\}_{\ell = 1}^{T'} \gets \{\mathbf{V}\mathbf{D}(t_{k_\ell})\mathbf{V}^{-1}\}_{\ell = 1}^{T'}$}
    \State{$\{\mathbf{\bar{Q}}_{t_{k_\ell}|t_{k_{\ell - 1}}}\}_{\ell = 1}^{T'} \gets \{\mathbf{V}\mathbf{D}(t_{k_{\ell - 1}})^{-1}\mathbf{D}(t_{k_\ell})\mathbf{V}^{-1}\}_{\ell = 1}^{T'}$}
    \State{Compute loss $\mathcal{L}_{\textsf{P2}}$ referring \eqref{eq:bcd loss function}}
    \State{Update $\{\mathsf{diag}(\mathbf{D}(t_{k_\ell}))\}_{\ell = 1}^{T'}$ via Adam optimizer}
    \State{$\{\mathbf{D}(t_{k_\ell})\}_{\ell = 1}^{T'} \gets \{\mathsf{ReLU}(\mathbf{D}(t_{k_\ell}))\}_{\ell = 1}^{T'}$}\label{line:enforce}\Comment{\textit{enforce}}
    \State{Apply PAVA on $\{\mathsf{diag}(\mathbf{D}(t_{k_\ell}))\}_{\ell = 0}^{T'}$\label{line:non-increasing}}\Comment{\textit{non-increasing}}
\EndWhile
\State{\textbf{Output:} Optimized $\mathbf{V}$ and $\{\mathbf{D}(t_{k_\ell})\}_{\ell = 1}^{T'}$}
\end{algorithmic}
\end{algorithm}

\begin{figure}[!t]
    \centering
    \includegraphics[width=0.8\linewidth]{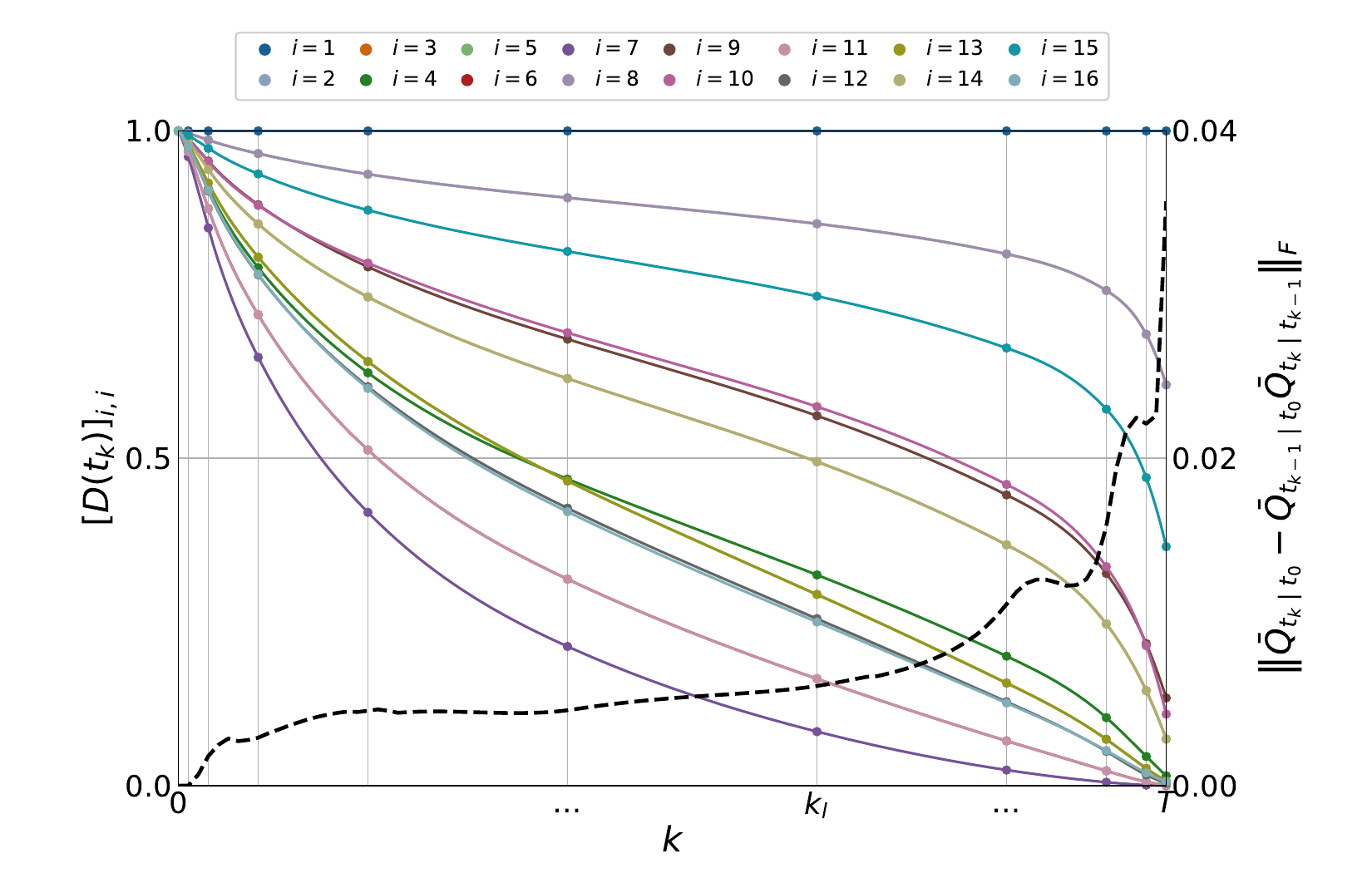}
    \caption{Diagonal entries of $\{\mathbf{D}(t_k)\}_{k = 0}^T$ obtained via spline interpolation from the optimized subsampled matrices $\{\mathbf{D}(t_{k_\ell})\}_{\ell = 1}^{T'}$ and Chapman-Kolmogorov consistency errors $\{\|\bar{\mathbf{Q}}_{t_k|t_0} - \bar{\mathbf{Q}}_{t_{k - 1}|t_0}\bar{\mathbf{Q}}_{t_k|t_{k - 1}}\|_\mathsf{F}\}_{k = 1}^T$ after ReLU and row-normalization, shown by the dashed line, when $M = 16$.}
    \label{fig:spline_D}
\end{figure}

To address the original problem \textsf{P1}, once the coarse-grained transition matrices $\{\mathbf{\bar{Q}}_{t_{k_\ell}|t_0}\}_{\ell = 1}^{T'}$ are obtained by solving \textsf{P2}, the remaining intermediate transition matrices are constructed by applying cubic spline interpolation\footnote{Although other interpolation methods, such as monotone or linear interpolation, are also applicable, we observed similar transition-matrix fitting and reconstruction performance across different interpolation schemes. Therefore, we adopt cubic spline interpolation because it provides smooth interpolation of the diagonal trajectories.} \cite{mckinley1998cubic} to the diagonal matrices $\{\mathbf{D}(t_{k_\ell})\}_{\ell = 0}^{T'}$, as shown in Fig. \ref{fig:spline_D}. This interpolation of $\{\mathbf{D}(t_{k_\ell})\}_{\ell = 0}^{T'}$ ensures that the approximated transition matrices $\{\mathbf{\bar{Q}}_{t_k|t_0}\}_{k = 1}^{T}$ preserve the eigen-decomposition structure $\mathbf{\bar{Q}}_{t_k|t_0} = \mathbf{V}\mathbf{D}(t_k)\mathbf{V}^{-1}$ in \eqref{eq:p0c1}, where the eigenmatrix $\mathbf{V}$ is shared across time points $k\in[1{:}T]$. In contrast, directly interpolating $\{\mathbf{\bar{Q}}_{t_{k_\ell}|t_0}\}_{\ell = 1}^{T'}$ cannot guarantee this shared-eigenbasis structure. The use of diagonal interpolation also allows the non-increasing property in \textbf{Lemma \ref{lm:non-increasing diagonal}} to be directly controlled along the diagonal trajectories, which helps maintain a physically consistent evolution of the diffusion process over time.

Finally, the approximated transition matrices are obtained as $\mathbf{\bar{Q}}_{t_k|t_0} = \mathbf{V}\mathbf{D}(t_k)\mathbf{V}^{-1}$, and the corresponding single-step transition matrices are computed as $\mathbf{\bar{Q}}_{t_{k}|t_{k - 1}} = \mathbf{\bar{Q}}_{t_{k - 1}|t_0}^{-1}\mathbf{\bar{Q}}_{t_k|t_0} = \mathbf{V}\mathbf{D}(t_{k - 1})^{-1}\mathbf{D}(t_k)\mathbf{V}^{-1}$. Since small negative entries may arise due to numerical fitting and interpolation errors, $\mathsf{ReLU}(\bar{\mathbf{Q}}_{t_k|t_0})$ and $\mathsf{ReLU}(\bar{\mathbf{Q}}_{t_k|t_{k - 1}})$ followed by row-normalization are applied before using them in the discrete DM to ensure row-stochasticity. Although the interpolation, ReLU, and row-normalization are performed separately from \textbf{Algorithm \ref{alg:block coordinate descent algorithm}} and may therefore introduce residual Chapman-Kolmogorov inconsistency, Fig. \ref{fig:spline_D} shows that the resulting consistency errors remain sufficiently small across all diffusion steps. The resulting transition matrices are then used for DM training and sampling, which is detailed in Section \ref{sec:training and sampling algorithms}.

Certainly, although the approximated solution matrices for the discrete DM define a slightly altered symbol detection region, the receiver still performs symbol detection using the optimal region. Moreover, the optimized transition matrices closely approximate the true symbol detection matrices, resulting in superior reconstruction quality by enforcing the Markov property, as demonstrated in Section \ref{sec:experimental results}.

\begin{figure}[!t]
    \centering
    \includegraphics[width=0.85\linewidth]{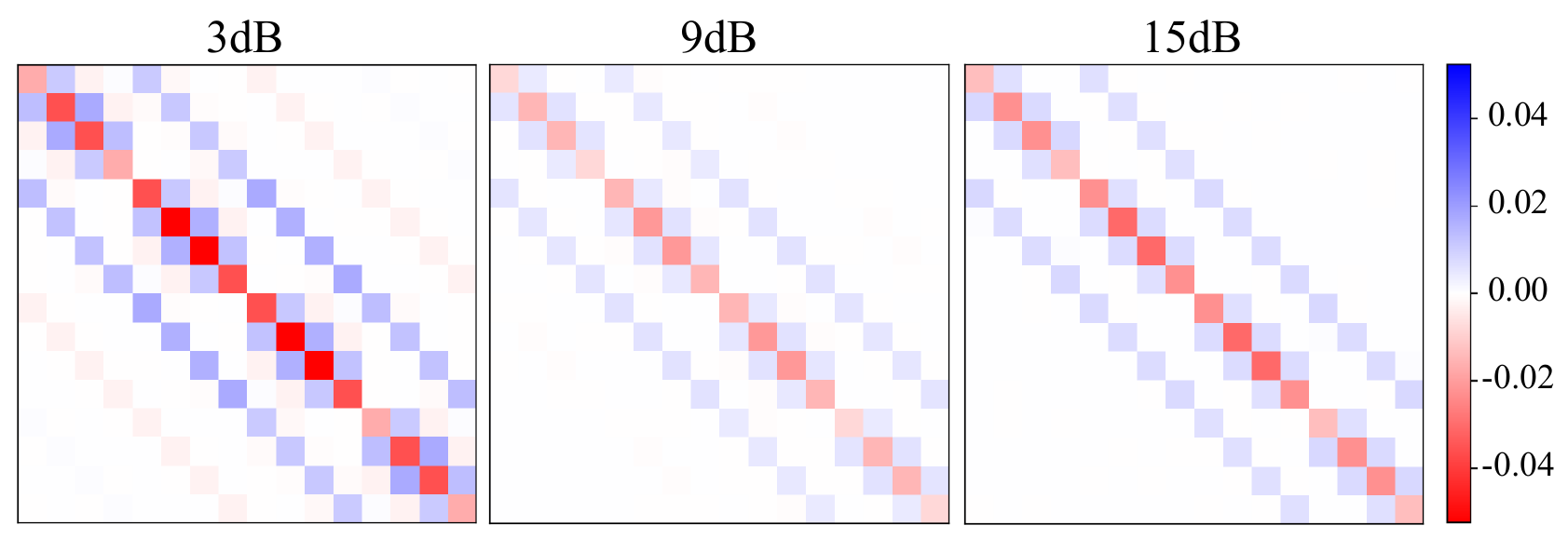}
    \caption{Numerically reconstructed rate matrices from the optimized transition matrices under different SNR values of 3dB, 9dB, and 15dB when $M = 16$.}
    \label{fig:rate matrix validation}
\end{figure}

\begin{remark}
Since \emph{\textsf{P1}} is a tractable reduction of the CTMC-constrained problem \emph{\textsf{P0}}, it does not theoretically guarantee all CTMC generator conditions in \emph{\textbf{Proposition \ref{prop:rate matrix properties}}}. Specifically, \emph{\textsf{P1}} ensures the row-stochasticity of the optimized transition matrices, but does not explicitly enforce the element-wise sign conditions of the corresponding rate matrix. Nevertheless, as shown in Fig. \ref{fig:rate matrix validation}, the rate matrices reconstructed from the optimized transition matrices, i.e., $\overrightarrow{\mathbf{R}_t} = \mathbf{V}\frac{\partial\mathsf{ln}(\mathbf{D}(\tau))}{\partial \tau}|_{\tau = t}\mathbf{V}^{-1}$, exhibit CTMC-compatible behavior in practice: their diagonal entries are clearly negative, and their off-diagonal entries are mostly non-negative or only slightly negative with small magnitudes. This numerical observation supports that the proposed solution provides an approximate Markov projection of the channel-induced transition dynamics, while the discrete DM itself directly operates on the optimized transition matrices rather than the rate matrices.
\end{remark}

\subsection{Semantics-Informed Diffusion Model Architecture}\label{subsec:semantics-informed diffusion model architecture}
For every reverse diffusion step, DMs typically employ the U-Net architecture \cite{ronneberger2015u} to obtain the conditional probability $p_{t_{k - 1}|t_k}^{\boldsymbol{\theta}}(\mathbf{u}_{t_{k - 1}}|\mathbf{u}_{t_k})\in\mathbb{R}^{N\times M}$, where $\mathbf{u}_{t_k}\in[1{:}M]^N$ is the corrupted sample consisting of the state indices $\{[\mathbf{u}_{t_k}]_i\}_{i = 1}^N$. In \cite{austin2021structured}, $\mathbf{u}_{t_k}$ is directly processed by the U-Net to compute the logits $\boldsymbol{\Phi}_{t_k}\in\mathbb{R}^{N\times M}$. A row-wise softmax is then applied as $p_{t_0|t_k}^{\boldsymbol{\theta}}(\mathbf{u}_{t_0}|\mathbf{u}_{t_k}) = \mathsf{Softmax}(\boldsymbol{\Phi}_{t_k})\in\mathbb{R}^{N\times M}$, which is subsequently used in \eqref{eq:learned reverse kernel} to compute $p_{t_{k - 1}|t_k}^{\boldsymbol{\theta}}$. However, $\mathbf{u}_{t_k}$ was treated purely as a sequence of categorical indices, despite each index corresponding to a learned codeword in the VQ codebook $\mathbf{C}\in\mathbb{R}^{M\times d}$. In this setting, therefore, the U-Net is trained without any explicit knowledge of the semantic structure of the states, limiting its ability to model relationships within the latent feature space.

To enable the DM to better estimate conditional probabilities by exploiting semantic information, we propose incorporating the pretrained VQ codebook $\mathbf{C}$ into the U-Net input. Specifically, before entering the U-Net, $\mathbf{u}_{t_k}$ is embedded through the given codebook as $\mathbf{W}_{t_k}\mathbf{C}\in\mathbb{R}^{N\times d}$, where $\mathbf{W}_{t_k}\in\mathbb{R}^{N\times M}$ is a matrix of row-wise one-hot vectors corresponding to $[\mathbf{u}_{t_k}]_i$ for $i\in[1{:}N]$. This embedding step injects the relational structure among the $N$ indices from the $M$ possible states, leveraging the semantic organization already learned in the VQ codebook. As a result, the U-Net can perform probability estimation in a more semantically informed manner, while reflecting the true digital communication channel behavior. Moreover, this embedding mechanism strengthens the integration between the DM and the JSCC structure, which is expected to yield a cooperative effect. Following \cite{austin2021structured}, the final conditional probability $p_{t_{k - 1}|t_k}^{\boldsymbol{\theta}}$ is obtained by applying the $\mathsf{Softmax}$ function to the output logits $\boldsymbol{\Phi}_{t_k}$ to compute $p_{t_0|t_k}^{\boldsymbol{\theta}}$, and then incorporating it into the reverse kernel modified from \eqref{eq:learned reverse kernel}, given by
\begin{align}\label{eq:modified learned reverse kernel}
    &p_{t_{k - 1}|t_k}^{\boldsymbol{\theta}}(u_{t_{k - 1}}|u_{t_k}, \mathbf{C})\\
    \nonumber&= \sum\nolimits_{u_{t_0} = 1}^M q_{t_{k - 1}|t_k, t_0}(u_{t_{k - 1}}|u_{t_k}, u_{t_0})p_{t_0|t_k}^{\boldsymbol{\theta}}(u_{t_0}|u_{t_k}, \mathbf{C}).
\end{align}

\section{Training and Inference Algorithms}\label{sec:training and sampling algorithms}
In this section, we present the training and inference algorithms for the SSCDM. The training procedure first pretrains the JSCC model without the SSCDM and then optimizes the SSCDM while keeping the JSCC model fixed. Furthermore, we introduce the novel loss function used for training the JSCC encoder-decoder and the VQ codebook. This proposed loss encourages the learned codebook to form a locally consistent topology in the latent space, ensuring that neighboring constellation symbols correspond to semantically related codewords. Then, the inference process at the receiver is subsequently described, where the SSCDM corrects symbol errors prior to the JSCC decoder.

\subsection{Training Method}\label{subsec:training method}
The proposed discrete DM-aided digital semantic communication is trained through a two-stage procedure that jointly optimizes the JSCC model parameters $\{\boldsymbol{\psi}_\mathsf{Tx}, \boldsymbol{\psi}_\mathsf{Rx}\}$, the VQ codebook $\mathbf{C}$, and the DM parameters $\boldsymbol{\theta}$. The training utilizes a dataset $\mathcal{D}$, approximated transition matrices $\{\mathbf{\bar{Q}}_{t_{k}|t_0}\}_{k = 1}^{T}$, the loss weights $\{\alpha, \beta, \gamma, \lambda\}$, and training SNR $\eta_{\mathsf{train}}$.

In the first stage, the VQ-based JSCC system without SSCDM is optimized in a mini-batch manner. For each iteration, mini-batch samples are transmitted through an AWGN channel with $\eta_\mathsf{train}$ employing the JSCC encoder $f_{\boldsymbol{\psi}_\mathsf{Tx}}$, decoder $g_{\boldsymbol{\psi}_\mathsf{Rx}}$, and learnable codebook $\mathbf{C}$. We adopt a reconstruction-based loss function provided in \cite{huh2025universal}:
\begin{align}\label{eq:image reconstruction loss}
\begin{split}
     \mathcal{L}_\mathsf{VQ} = \mathbb{E}_\mathbf{X}[\mathsf{MSE}(\hat{\mathbf{X}}, \mathbf{X}) &+ \alpha \mathsf{MSE}(\hat{\mathbf{Y}}, \mathsf{sg}(\mathbf{Y}))\\
     &+ \beta \mathsf{MSE}(\mathbf{Y}, \mathsf{sg}(\hat{\mathbf{Y}}))],
\end{split}
\end{align}
where $\mathsf{MSE}(\mathbf{A}, \mathbf{B})$ represents the mean squared error (MSE) between two inputs $\mathbf{A}$ and $\mathbf{B}$, and $\mathsf{sg}(\mathbf{A})$ is the stop-gradient operation that prevents the flow of gradients through $\mathbf{A}$ during backpropagation. Each component of the above loss serves a distinct function:
\begin{itemize}
    \item The first term ensures that the encoder and decoder reconstruct the input image accurately, accounting for distortions introduced by both the VQ and the channel noise. Since the VQ operation is non-differentiable, a straight-through estimator \cite{van2017neural} is employed to enable gradient propagation from $\mathbf{Y}$ to $\hat{\mathbf{Y}}$.
    \item The second term guides the optimization of the codebook by pulling the quantized output $\hat{\mathbf{Y}}$ closer to the encoder output $\mathbf{Y}$, which has been degraded by channel noise. The stop-gradient on $\mathbf{Y}$ prevents encoder parameters from being affected during this update.
    \item The third term acts as a commitment loss, promoting stability of the codebook convergence by encouraging the encoder output $\mathbf{Y}$ to stay close to its quantized counterpart $\hat{\mathbf{Y}}$. The stop-gradient on $\hat{\mathbf{Y}}$ ensures that only the encoder is updated, leaving the codebook unchanged during this step.
\end{itemize}
After the image reconstruction, $\mathcal{L}_\mathsf{VQ}$ is then computed, and $\{\boldsymbol{\psi}_\mathsf{Tx}, \boldsymbol{\psi}_\mathsf{Rx}, \mathbf{C}\}$ are updated via the Adam optimizer.

In the second stage, with the pretrained VQ-based JSCC parameters $\{\boldsymbol{\psi}_\mathsf{Tx}^*, \boldsymbol{\psi}_\mathsf{Rx}^*, \mathbf{C}^*\}$ fixed, the discrete DM parameters $\boldsymbol{\theta}$ are optimized. For each iteration, mini-batch data are encoded into index vectors using the encoder and VQ codebook, where each index vector $\mathbf{z}$ serves as the clean data sample $\mathbf{u}_{t_0}$ for the SSCDM training. After a time point $t_k$ is selected by sampling $k$ from a discrete uniform distribution $\mathsf{Unif}[1{:}T]$, $\mathbf{u}_{t_0}$ is randomly corrupted into $\mathbf{u}_{t_k}$ by the transition matrix $\mathbf{\bar{Q}}_{t_{k}|t_0}$, which captures the characteristics of digital communication channel while satisfying the Markov property. The perturbed sample $\mathbf{u}_{t_k}$ is subsequently refined into $\hat{\mathbf{u}}_{t_0}$ by $\mathsf{DM}_{\boldsymbol{\theta}}$, which first embeds $\mathbf{u}_{t_k}$ using $\mathbf{C}$ to exploit the geometry of the learned semantic latent space. The overall training objective of the SSCDM follows $\mathcal{L}_\mathsf{\lambda}$ in \eqref{eq:D3PM loss}, and the parameters $\boldsymbol{\theta}$ of $\mathsf{DM}_{\boldsymbol{\theta}}$ are updated using the Adam optimizer.

\subsection{SOM-Based Codebook Learning}\label{subsec:som-based codebook learning}
\begin{algorithm}[!t]
\caption{Training algorithm}\label{alg:training algorithm}
\begin{algorithmic}[1]
    \State{\textbf{Input: }Dataset $\mathcal{D}$, Markov-enforced transition matrices $\{\mathbf{\bar{Q}}_{t_{k}|t_0}\}_{k = 1}^{T}$, hyperparameters $\{\alpha, \beta, \gamma, \lambda\}$, training SNR $\eta_{\mathsf{train}}$}
    \State{Initialize training parameters $\{\boldsymbol{\psi}_\mathsf{Tx}, \boldsymbol{\psi}_\mathsf{Rx}, \mathbf{C}, \boldsymbol{\theta}\}$}
    \Statex\textit{\textbf{Stage 1: Train JSCC without SSCDM}}
    \While{not converged}
        \State{Sample $\mathbf{X}\in\mathcal{D}$}\Comment{\textit{mini-batch data}}
        \State{Encode $\mathbf{X}$ into $\mathbf{Y}$ using $f_{\boldsymbol{\psi}_\mathsf{Tx}}$}
        \State{Quantize $\mathbf{Y}$ into $\mathbf{z}$ using $\mathbf{C}$ and modulate into $\mathbf{s}$}
        \State{Receive $\hat{\mathbf{s}}$ via AWGN channel with $\eta_\mathsf{train}$}
        \State{Demodulate $\hat{\mathbf{s}}$ into $\tilde{\mathbf{z}}$ and dequantize into $\hat{\mathbf{Y}}$ using $\mathbf{C}$}
        \State{Decode $\hat{\mathbf{Y}}$ into $\hat{\mathbf{X}}$ using $g_{\boldsymbol{\psi}_\mathsf{Rx}}$}
        \State{Calculate $\mathcal{L}_\mathsf{VQ\text{-}SOM}$ as \eqref{eq:total loss}}
        \State{Update $\{\boldsymbol{\psi}_\mathsf{Tx}, \boldsymbol{\psi}_\mathsf{Rx}, \mathbf{C}\}$ via Adam optimizer}
    \EndWhile
    \Statex\textit{\textbf{Stage 2: Train SSCDM for symbol error correction}}
    \State{Freeze $\{\boldsymbol{\psi}_\mathsf{Tx}, \boldsymbol{\psi}_\mathsf{Rx}, \mathbf{C}\}$}
    \While{not converged}
        \State{Sample $\mathbf{X}\in\mathcal{D}$\Comment{\textit{mini-batch data}}}
        \State{Get $\mathbf{z}$ from $\mathbf{X}$ using $f_{\boldsymbol{\psi}_\mathsf{Tx}}$ and $\mathbf{C}$}
        \State{Define $\mathbf{u}_{t_0} = \mathbf{z}$}
        \State{Sample $k \sim \mathsf{Unif}[1{:}T]$}
        \State{Perturb $\mathbf{u}_{t_0}$ into $\mathbf{u}_{t_k}$ using $\mathbf{\bar{Q}}_{t_{k}|t_0}$}
        \State{Correct $\mathbf{u}_{t_k}$ into $\hat{\mathbf{u}}_{t_0}$ using $\mathsf{DM}_{\boldsymbol{\theta}}$}
        \State{Calculate $\mathcal{L}_\mathsf{\lambda}$ as \eqref{eq:D3PM loss}}
        \State{Update $\boldsymbol{\theta}$ via Adam optimizer}
    \EndWhile
    \State{\textbf{Output: }Optimized parameters $\{\boldsymbol{\psi}_\mathsf{Tx}^*, \boldsymbol{\psi}_\mathsf{Rx}^*, \mathbf{C}^*, \boldsymbol{\theta}^*\}$}
\end{algorithmic}
\end{algorithm}

In conventional digital modulation, bit sequences are typically Gray-mapped to minimize the Hamming distance between adjacent symbols. However, a VQ codebook trained only with the reconstruction loss in \eqref{eq:image reconstruction loss} does not inherently preserve such locality in the Euclidean feature space. That is, the codewords corresponding to two neighboring digital symbols may not be closer to each other than those of more distant symbols, unlike the property ensured by Gray-mapping. As a result, although our SSCDM attempts to pull a misdetected symbol index back toward its original one, the corresponding codewords may still be far apart in the feature space. This discrepancy causes a mismatch between the correction capability and the image reconstruction quality, since correcting a symbol index does not necessarily yield a semantically similar codeword to the original one.

To alleviate this issue, we incorporate the concept of SOM \cite{kohonen2002self}. Recall that the second term in \eqref{eq:image reconstruction loss}, i.e., $\mathsf{MSE}(\hat{\mathbf{Y}}, \mathsf{sg}(\mathbf{Y)})$, encourages each received codeword $\mathbf{c}_{\tilde{z}_i}$ to move closer to its corresponding feature vector $\mathbf{y}_i$. The SOM mechanism extends this idea by also attracting the neighboring codewords $\mathbb{N}(\tilde{z}_i)$ on the constellation map toward $\mathbf{y}_i$. Here, $\mathbb{N}(\tilde{z}_i)$ denotes the set of manually defined neighboring symbol indices based on the topology of $\mathbf{s}_{\tilde{z}_i}$. This regularization is implemented by
\begin{align}\label{eq:som loss}
    \mathcal{L}_\mathsf{SOM} = \mathbb{E}_\mathbf{X}\left[\sum\nolimits_{i = 1}^N\sum\nolimits_{j\in\mathbb{N}(\tilde{z}_i)}\lambda_{ij}\mathsf{MSE}(\mathbf{c}_{j}, \mathsf{sg}(\mathbf{y}_i))\right],
\end{align}
where $\lambda_{ij}$ is a weighting coefficient inversely proportional to the Euclidean distance between symbols $\mathbf{s}_i$ and $\mathbf{s}_j$.

Following \cite{fortuin2018som}, we define $\mathbb{N}(\tilde{z}_i)$ as the four nearest neighbors (up, down, left, and right) on the constellation map, and set $\lambda_{ij} = 1$. When fewer than four neighbors exist, only the available ones are used. Finally, the total training objective for the JSCC encoder–decoder and the VQ codebook is given by
\begin{align}\label{eq:total loss}
    \mathcal{L}_\mathsf{VQ\text{-}SOM} = \mathcal{L}_\mathsf{VQ} + \gamma\mathcal{L}_\mathsf{SOM},
\end{align}
where $\gamma$ is a balancing coefficient between $\mathcal{L}_\mathsf{VQ}$ and $\mathcal{L}_\mathsf{SOM}$. By jointly updating the codewords using the SOM regularization, the learned codebook forms a smoother and more topology-preserving manifold in the feature space, analogous to the locality-preserving property of Gray-mapping. The entire training process\footnote{Despite the fixed training SNR in the first stage, the SOM-based codebook learning enhances the channel robustness of the JSCC model in a manner analogous to Gray mapping. In the second stage, SSCDM is trained over varying SNR values for symbol correction. Taken together, these two stages yield an overall training procedure that accounts for various SNR conditions.} is summarized in \textbf{Algorithm \ref{alg:training algorithm}}, while the first stage employs $\mathcal{L}_\mathsf{VQ\text{-}SOM}$ in \eqref{eq:total loss} instead of $\mathcal{L}_\mathsf{VQ}$ in \eqref{eq:image reconstruction loss}.

\subsection{Decoding Process with SSCDM}\label{subsec:decoding process with sscdm}
\begin{algorithm}[!t]
\caption{Receiver inference algorithm}\label{alg:receiver inference algorithm}
\begin{algorithmic}[1]
    \State{\textbf{Input: }$\mathsf{DM}_{\boldsymbol{\theta}}$, $g_{\boldsymbol{\psi}_\mathsf{Rx}}$, $\mathbf{C}$, SNR $\eta$, received index vector $\tilde{\mathbf{z}}$}
    \State{Find the correction starting index $k^*$ from $\eta$ as \eqref{eq:starting point selection}}\label{line:dm start}
    \State{Define $\mathbf{u}_{t_{k^*}} = \tilde{\mathbf{z}}$}
    \For{$k = {k^*}, {k^*-1}, \dots, {2}$}
        \State{Compute $p_{t_{k - 1}|t_k}^{\boldsymbol{\theta}}(\mathbf{u}_{t_{k - 1}}|\mathbf{u}_{t_k}, \mathbf{C})$ using $\mathsf{DM}_{\boldsymbol{\theta}}$ as \eqref{eq:modified learned reverse kernel}}
        \State{Sample $\mathbf{u}_{t_{k - 1}}$ from $p_{t_{k - 1}|t_k}^{\boldsymbol{\theta}}$}
    \EndFor
    \State{Compute $p_{t_0|t_1}^{\boldsymbol{\theta}}(\mathbf{u}_{t_0}|\mathbf{u}_{t_1}, \mathbf{C})$ using $\mathsf{DM}_{\boldsymbol{\theta}}$}
    \State{$\{[\mathbf{u}_{t_0}]_{i}\}_{i = 1}^N \gets \{\argmax_{j} p_{t_0|t_1}^{\boldsymbol{\theta}}([\mathbf{u}_{t_0}]_i = j|\mathbf{u}_{t_1}, \mathbf{C})\}_{i = 1}^N$}
    \State{$\hat{\mathbf{z}} \gets \mathbf{u}_{t_0}$}\label{line:dm end}
    \State{Dequantize $\hat{\mathbf{z}}$ into $\hat{\mathbf{Y}}$ using $\mathbf{C}$}
    \State{Decode $\hat{\mathbf{Y}}$ into $\hat{\mathbf{X}}$ using $g_{\boldsymbol{\psi}_\mathsf{Rx}}$}
    \State{\textbf{Output: }Reconstructed image $\hat{\mathbf{X}}$}
\end{algorithmic}
\end{algorithm}

After the training stages, the receiver employs the SSCDM, i.e., $\mathsf{DM}_{\boldsymbol{\theta}}$, to correct symbol errors prior to image reconstruction, where this whole procedure is presented in \textbf{Algorithm \ref{alg:receiver inference algorithm}}. Given the channel SNR $\eta$ and the predefined noise schedule $\{\sigma_{t_{k}}^2\}_{k = 1}^{T}$, the starting time index $k^*$ of the reverse diffusion process is determined as
\begin{align}\label{eq:starting point selection}
    k^* = \argmin_{k\in[1:T]} \left|\eta - P/\bar{\sigma}_{t_k}^2\right|,
\end{align}
where $\bar{\sigma}_{t_k}^2 = \sum_{k' = 1}^k\sigma_{t_{k'}}^2$. The current channel condition thus determines the starting time point $t_{k^*}$ of the reverse diffusion process, where higher SNR corresponds to $t_{k^*}$ closer to $0$ and lower SNR is mapped to $t_{k^*}$ closer to $1$.

The received index vector $\tilde{\mathbf{z}}$ is then regarded as the perturbed sample $\mathbf{u}_{t_{k^*}}$ in the discrete DM. Starting from $t_{k^*}$, the reverse diffusion proceeds iteratively down to $t_1$. At each step $k$, the reverse kernel $p_{t_{k-1}|t_k}^{\boldsymbol{\theta}}(\mathbf{u}_{t_{k-1}}|\mathbf{u}_{t_k}, \mathbf{C})$ is computed by $\mathsf{DM}_{\boldsymbol{\theta}}$ as defined in \eqref{eq:modified learned reverse kernel}, and $\mathbf{u}_{t_{k-1}}$ is sampled accordingly. Note that, when computing $p_{t_{k-1}|t_k}$, $\mathbf{u}_{t_k}$ is projected onto the semantic latent space by replacing it with the corresponding codeword from the learned codebook $\mathbf{C}$. To ensure stable reconstruction, once the process reaches $t_1$, the final clean data $\mathbf{u}_{t_0}$ is obtained by computing $p_{t_0|t_1}^{\boldsymbol{\theta}}(\mathbf{u}_{t_0}|\mathbf{u}_{t_1}, \mathbf{C})$ via $\mathsf{DM}_{\boldsymbol{\theta}}$ and performing maximum-likelihood symbol selection for each entry, rather than stochastic sampling. The corrected index vector is thus given by $\hat{\mathbf{z}} = \mathbf{u}_{t_0}$. The above operations correspond to lines \ref{line:dm start}-\ref{line:dm end} of \textbf{Algorithm \ref{alg:receiver inference algorithm}}, i.e., $\hat{\mathbf{z}} = \mathsf{DM}_{\boldsymbol{\theta}}(\tilde{\mathbf{z}}, \eta, \mathbf{C})$ in Section \ref{sec:system model}. Finally, $\hat{\mathbf{z}}$ is dequantized into the discretized feature representation $\hat{\mathbf{Y}}$ using the codebook $\mathbf{C}$, and the JSCC decoder $g_{\boldsymbol{\psi}_\mathsf{Rx}}$ reconstructs the image $\hat{\mathbf{X}}$.

\section{Experimental Results}\label{sec:experimental results}
In this section, we evaluate the performance of the proposed SSCDM by employing two standard image reconstruction metrics: the multi-scale structural similarity index measure (MS-SSIM) and the learned perceptual image patch similarity (LPIPS), where higher MS-SSIM and lower LPIPS indicate better image quality. Experiments are conducted on the FFHQ dataset \cite{karras2019style}, which consists of color images with a high resolution of $256\times256$. The NN architecture follows the residual convolutional neural network (CNN)-based design in \cite{huh2025universal}, adapted for the FFHQ dataset by setting $c_1 = 128$ and $c_2 = 256$, with both the encoder and decoder consisting of five layers. This configuration results in a compression ratio of $\frac{N}{CHW} = \frac{1}{192}$.

We assume 16-QAM modulation\footnote{Our method is not restricted to a specific modulation order and can be directly extended to other constellations by adjusting the codebook size.} for transmission, leading to a codebook size of $16$, and each codeword is represented as a 4-dimensional vector. The batch size is set to 32, and the learning rate is initialized to $0.01$ with a step decay scheduler, where the decay rate is $0.5$ and the decay period is 80 epochs. The total number of training epochs is 400, and the training SNR is fixed at $\eta_{\mathsf{train}} = 20$dB to ensure sufficiently reliable symbol observations, while the DM is applied under lower SNR regimes. For optimizing the JSCC encoder–decoder and the VQ codebook in the first training stage, the hyperparameters in $\mathcal{L}_{\mathsf{VQ\text{-}SOM}}$ are configured as $\alpha = 1$, $\beta = 0.25$, and $\gamma = 0.9$.

To optimize SSCDM, a set of Markov-enforced matrices $\{\mathbf{Q}_{t_k|t_0}\}_{k = 1}^T$ with $t_k = k/T$ must be predetermined, which are obtained through \textbf{Algorithm \ref{alg:block coordinate descent algorithm}} and then the cubic spline interpolation. The hyperparameters in the loss function \eqref{eq:bcd loss function} are $\lambda_{1} = \lambda_{3} = 10$ and $\lambda_{2} = 0.001$. We set the maximum diffusion step to $T = 100$, and the noise variance sequence $\{\sigma_{t_k}^2\}_{k = 1}^T$ is designed using the sigmoid-based scheduling method described in Appendix \ref{apdx:noise scheduling method}, covering the SNR range from $-3$dB to $15$dB. For the subsequence $\{t_{k_\ell}\}_{\ell = 1}^{T'}$ with $T' = 10$, the time steps are manually chosen to be denser near the beginning and the end of the full sequence $\{t_k\}_{k = 1}^T$. Specifically, we set $\{k_\ell\}_{\ell = 1}^{T'} = \{2, 4, 9, 20, 40, 65, 84, 94, 98, 100\}$ as indicated by the vertical lines in Fig. \ref{fig:spline_D}.

For DM-based digital symbol error correction modules, we adopt U-Net as the backbone architecture, commonly employed for DMs such as DDPM \cite{ho2020denoising} and D3PM \cite{austin2021structured}. The DM input corresponds to the feature map composed of the VQ codeword indices from the symbol detection. This input is first expanded to a base tensor channel dimension $\kappa$ through an input layer, followed by three down-sampling and up-sampling stages. During these stages, self-attention is applied at the resolution of $16\times16$. In the second training stage, the DM is optimized via $\mathcal{L}_{\mathsf{\lambda}}$ with $\lambda = 0.001$. The batch size and learning rate are set to $32$ and $2\times10^{-4}$, respectively, and the total training iteration is $400,000$.

\begin{figure*}[!t]
    \centering
    \subfigure[Heatmap]{\includegraphics[width=0.67\linewidth]{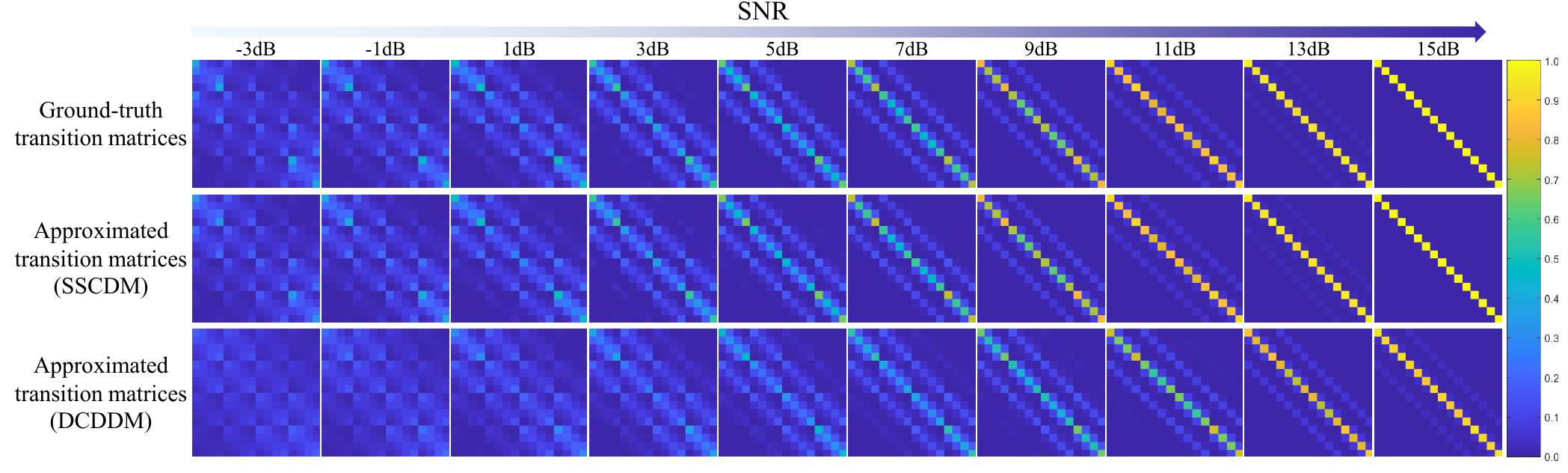}
    \label{fig:transition_matrix_heatmap}}
    \subfigure[NMSE]{\includegraphics[width=0.31\linewidth]{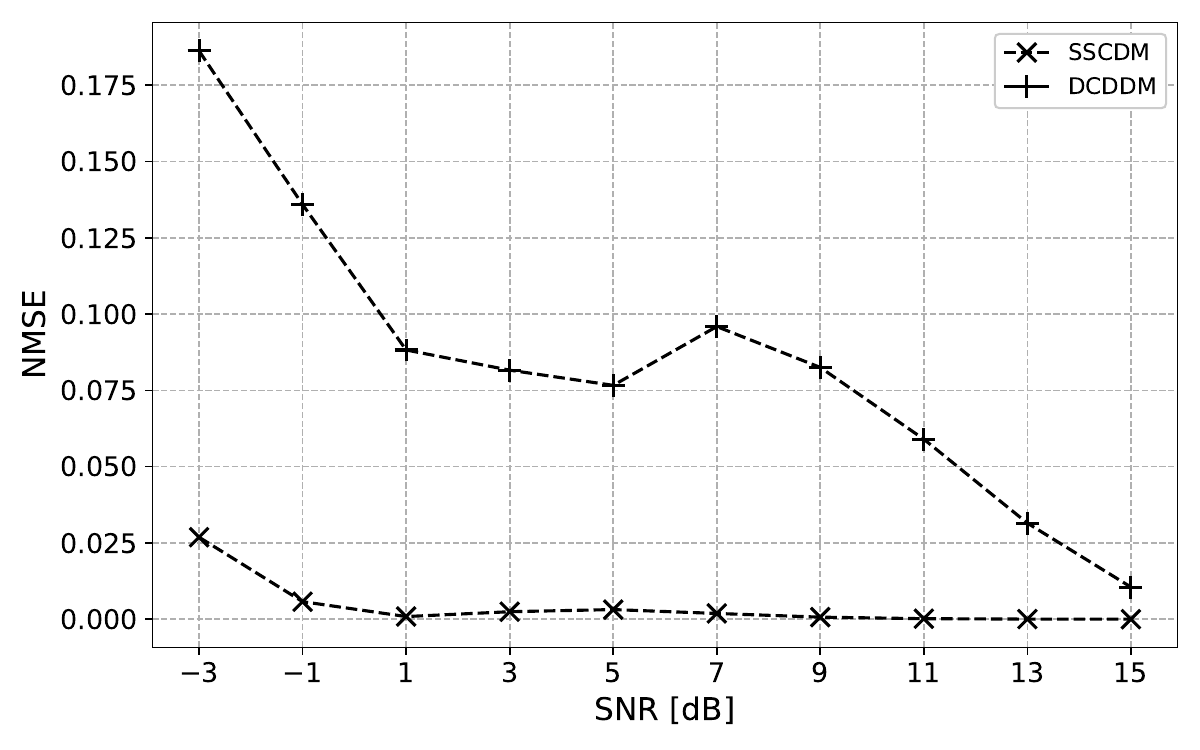}
    \label{fig:transition_matrix_nmse}}
    \caption{(a) Heatmaps of 16-QAM transition matrix for the ground-truth, SSCDM, and DCDDM under different SNR conditions and (b) NMSE of the designed transition matrices of SSCDM and DCDDM.}
    \label{fig:transition_matrix}
\end{figure*}

\begin{figure*}[!t]
    \centering
    \subfigure[FFHQ]{\includegraphics[width=0.32\linewidth]{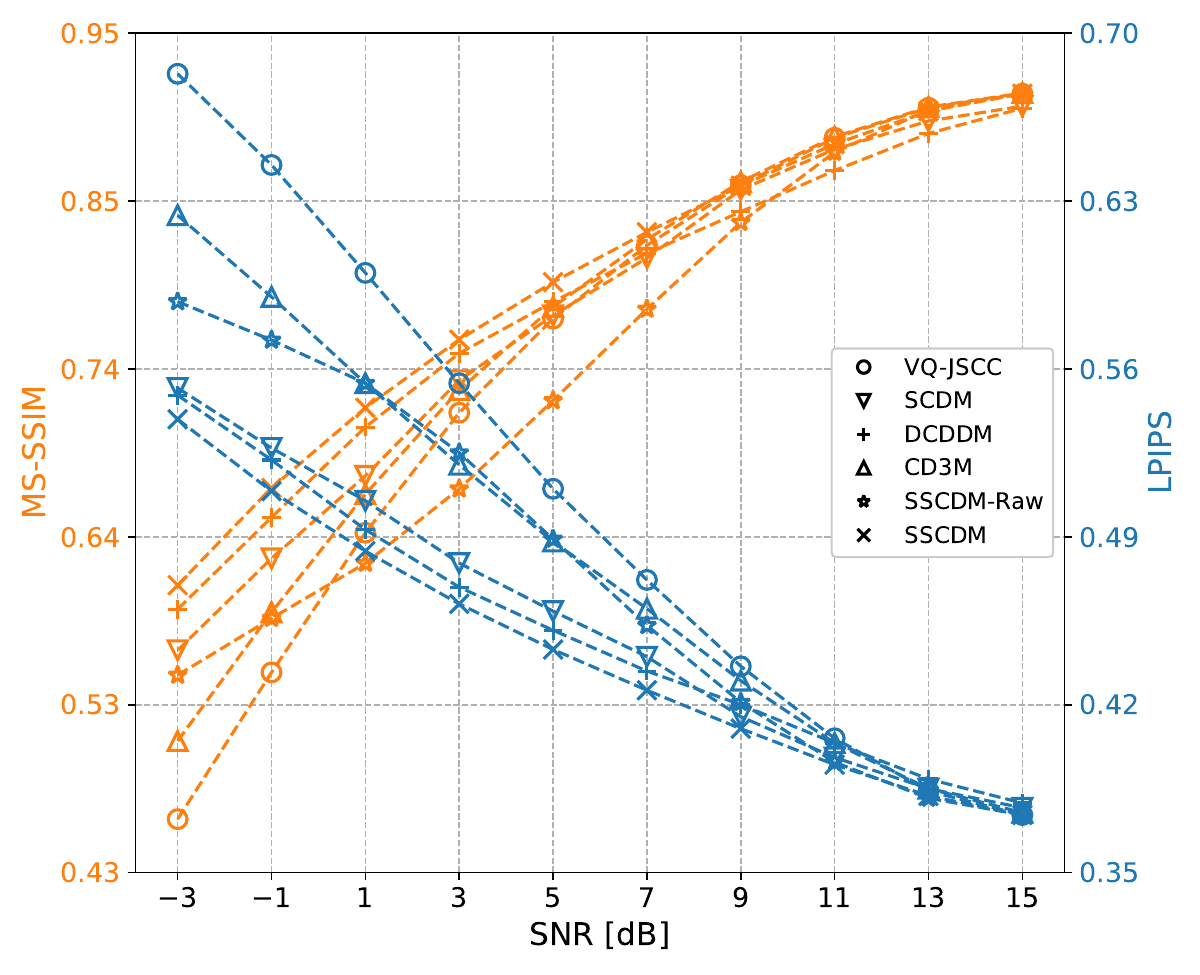}
    \label{fig:ffhq_result}}
    \subfigure[CelebA]{\includegraphics[width=0.32\linewidth]{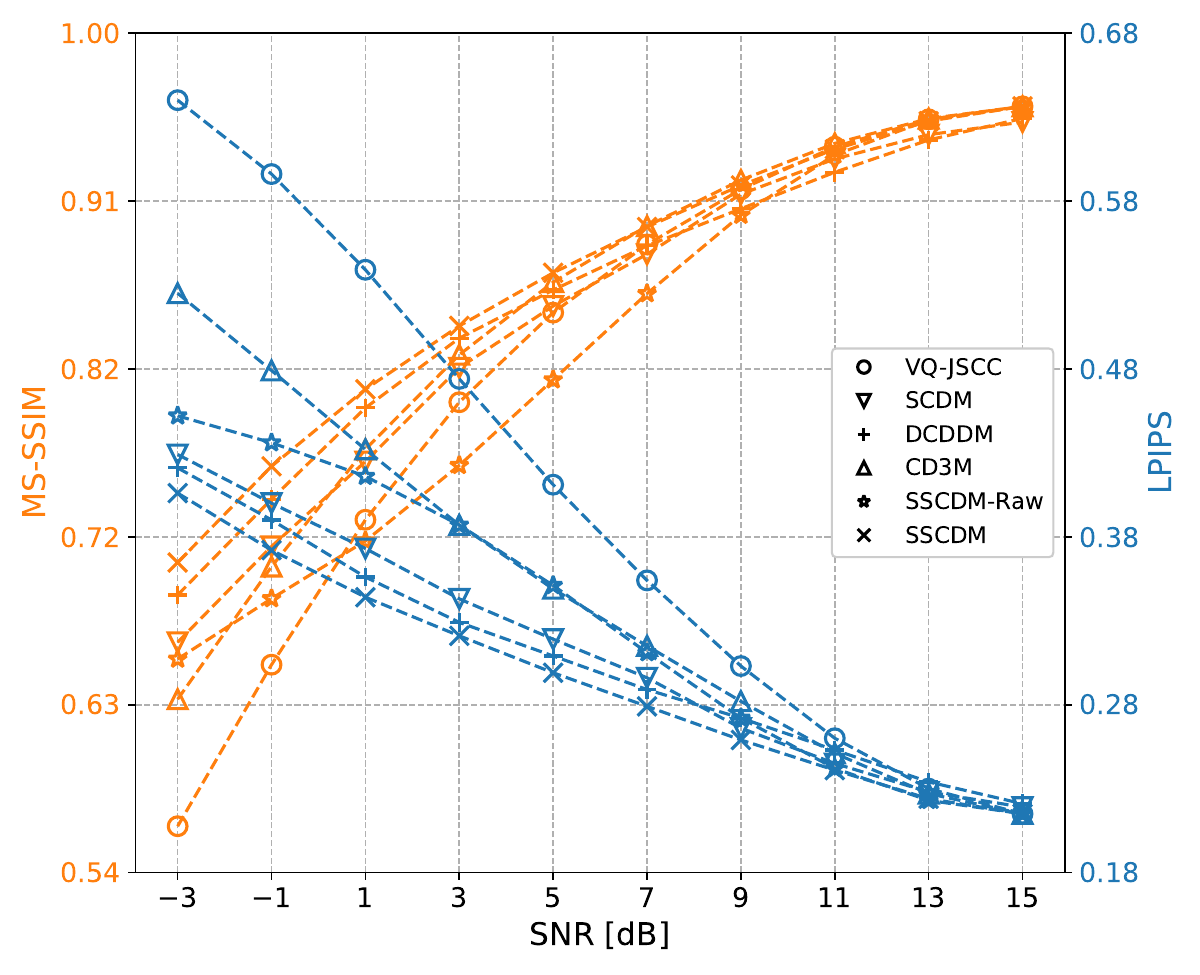}
    \label{fig:celeba_result}}
    \subfigure[CelebA-HQ]{\includegraphics[width=0.32\linewidth]{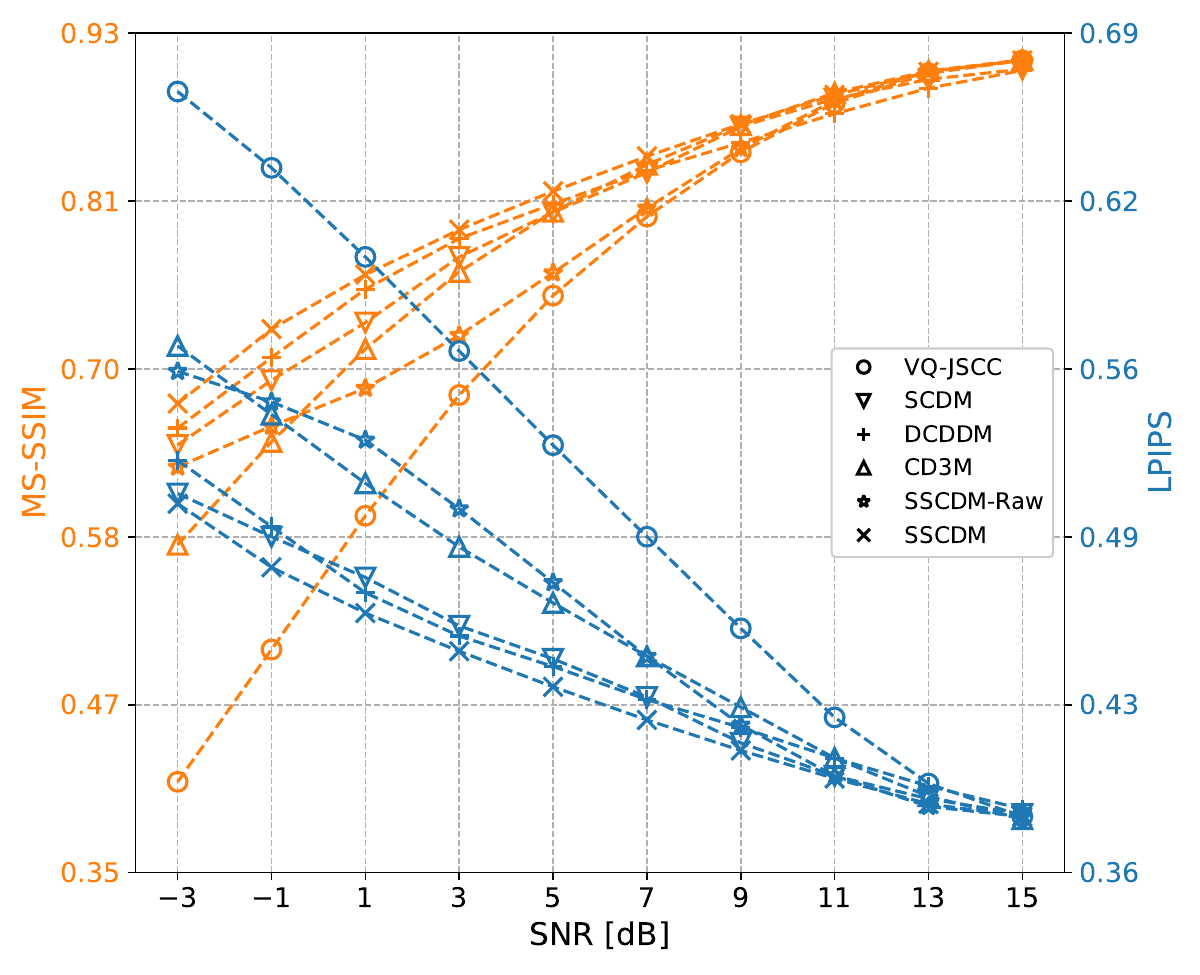}
    \label{fig:celebahq_result}}
    \caption{MS-SSIM$\uparrow$ (orange) and LPIPS$\downarrow$ (blue) comparison between SSCDM and baseline methods on (a) FFHQ, (b) CelebA, and (c) CelebA-HQ datasets.}
    \label{fig:comparison_results}
\end{figure*}

For the comparative analysis, we consider the following baselines. The JSCC encoder-decoder and the VQ codebook for all the baselines are trained using $\mathcal{L}_\mathsf{VQ\text{-}SOM}$ defined in \eqref{eq:total loss}.
\begin{itemize}
    \item \textbf{VQ-JSCC:} This method serves as the baseline model without the DM module. It is trained solely in the first stage of \textbf{Algorithm \ref{alg:training algorithm}}.

    \item \textbf{CD3M \cite{tang2025cd3m}:} This method applies D3PM \cite{austin2021structured} for symbol-level error correction in digital semantic communication. CD3M estimates per-step transition matrices independently by fitting a step-wise transition matrix between consecutive symbol detection matrices. The training configurations largely follow those reported in \cite{tang2025cd3m}. For a fair comparison, the scale factor of the noise scheduler is set to 0.9976 to match the target SNR range, with an identical model architecture to ours.

    \item \textbf{DCDDM \cite{he2025residual}:} This method also utilizes D3PM \cite{austin2021structured} for discrete DM to perform digital symbol correction. Each transition matrix includes a diagonal term representing correct transmission, while the remaining probability is allocated to other symbols in proportion to their pairwise symbol error probabilities. The overall model architecture and training settings are aligned with those of SSCDM. For the unspecified noise variance sequence in \cite{he2025residual}, we adopt a commonly used linear scheduling strategy.

    \item \textbf{SCDM \cite{mo2025scdm}:} This method adopts the score-based DM \cite{song2020score}, where continuous diffusion is applied to denoise AWGN-corrupted digital symbols before symbol detection. We follow the training procedure described in \cite{mo2025scdm}, while the size of DM in \cite{song2020score} is adjusted to match that of other methods for a fair comparison.

    \item \textbf{SSCDM-Raw:} This method is an ablation baseline of SSCDM with the same architecture and training configuration, but directly adopts the raw channel-induced transition matrices in \textbf{Proposition \ref{prop:transition matrix condition}} instead of the proposed Markov-enforced transition matrices.
\end{itemize}

\subsection{Transition Matrix Similarity}\label{subsec:transition matrix similarity}
Fig. \ref{fig:transition_matrix_heatmap} presents the heatmaps of 16-QAM transition matrices under varying SNR conditions. The proposed SSCDM successfully preserves the Markov property while accurately reflecting the characteristics of digital communication channel, as evidenced by the strong visual similarity between the ground-truth transition matrices (first row) and those by SSCDM (second row). In contrast, DCDDM (third row) exhibits clear discrepancies from the ground truth, mainly due to the absence of a proper Markov-chain formulation and inadequate modeling of channel dynamics.

To quantitatively assess the similarity between transition matrices, we compute the normalized mean square error (NMSE), i.e., $||\mathbf{Q}_{t_k|t_0} - \bar{\mathbf{Q}}_{t_k|t_0}||_F^2/|| \mathbf{Q}_{t_k|t_0}||_F^2$, where $\mathbf{Q}_{t_k|t_0}$ denotes the ground-truth transition matrix and $\bar{\mathbf{Q}}_{t_k|t_0}$ represents the obtained one from SSCDM. In Fig. \ref{fig:transition_matrix_nmse}, the NMSE of SSCDM converges to nearly zero across all SNR regions, confirming its ability to precisely capture the true channel transition behavior. In contrast, the approximated matrices from DCDDM exhibit noticeable errors relative to the ground truth. These results collectively validate the effectiveness of the proposed optimization method in \textbf{Algorithm \ref{alg:block coordinate descent algorithm}}. Consequently, SSCDM successfully models the underlying channel dynamics while preserving the Markov property, thereby enabling more stable and reliable diffusion-based correction.

\subsection{Comparison with Baselines}\label{subsec:comparison with baselines}

\begin{figure*}[!t]
    \centering
    \includegraphics[width=0.8\linewidth]{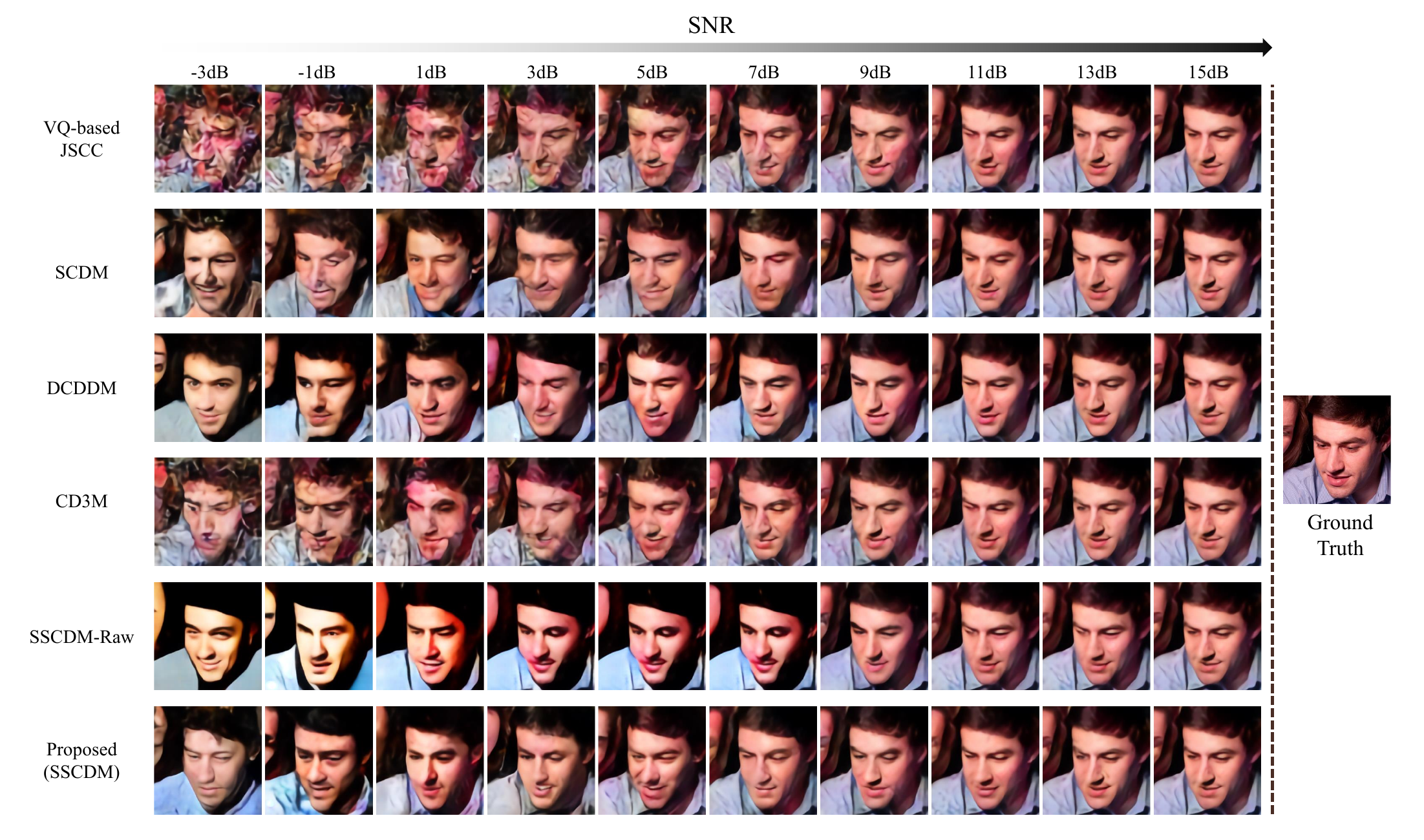}
    \caption{Reconstructed images from SSCDM and baseline methods with varying SNRs.}
    \label{fig:denoised_ffhq_img}
\end{figure*}

As shown in Fig. \ref{fig:ffhq_result}, our method consistently surpasses all baselines across the full SNR range, with particularly pronounced gains in low-SNR regions. While other DM-based schemes, except SSCDM-Raw, also enhance performance compared to VQ-JSCC, SSCDM delivers a more significant advantage by constructing transition matrices that both reflect the characteristics of digital communication channels and enforce the Markov property required by the DM. This is further supported by SSCDM-Raw, which directly uses the raw channel-induced transition matrices in \textbf{Proposition \ref{prop:transition matrix condition}} without addressing their non-Markovian nature. SSCDM-Raw exhibits degraded reconstruction performance, especially in the mid-SNR range, and even falls below VQ-JSCC in terms of MS-SSIM. As the SER drops sharply in high-SNR regions, both SCDM and DCDDM are slightly inferior to VQ-JSCC. This indicates that the continuous DM employed in SCDM is less suitable for AWGN denoising in digital communication, whereas the heuristic transition matrices of DCDDM fail to represent discrete diffusion precisely for symbol correction. In contrast, SSCDM maintains comparable performance to VQ-JSCC even at high SNRs, owing to our principled design of transition matrices. While CD3M attains competitive performance in high-SNR conditions, its gains in low-SNR regimes remain limited compared to other methods, primarily because it directly employs symbol detection matrices that are non-Markovian in typical digital communication.

Furthermore, as illustrated in Fig. \ref{fig:denoised_ffhq_img}, SSCDM delivers more stable and coherent reconstructions, whereas other schemes often exhibit degraded image quality with noticeable artifacts or structural distortions. Consequently, our method yields smoother and more perceptually faithful reconstructions in both quantitative and qualitative evaluations. To further validate the robustness of SSCDM across different datasets and image resolutions, we conduct additional experiments on the CelebA and CelebA-HQ datasets \cite{liu2015faceattributes, karras2018progressive}, consisting of $128\times128$ and $512\times512$ color images, respectively. For CelebA, the model configuration is adapted from the FFHQ experiments by setting $c_1 = 64$ and $c_2 = 128$, with both the encoder and decoder composed of four layers, resulting in a compression ratio of $\frac{N}{CHW} = \frac{1}{48}$. For CelebA-HQ, the FFHQ model configuration is retained. As depicted in Figs. \ref{fig:celeba_result} and \ref{fig:celebahq_result}, SSCDM exhibits performance trends similar to those observed on FFHQ, consistently outperforming other methods across all SNR levels. These results demonstrate that SSCDM maintains stable performance for different image resolutions.

\subsection{Performance under Various Communication Settings}\label{subsec:performance under different communication settings}

\begin{figure*}[!t]
    \centering
    \subfigure[Rayleigh fading]{\includegraphics[width=0.32\linewidth]{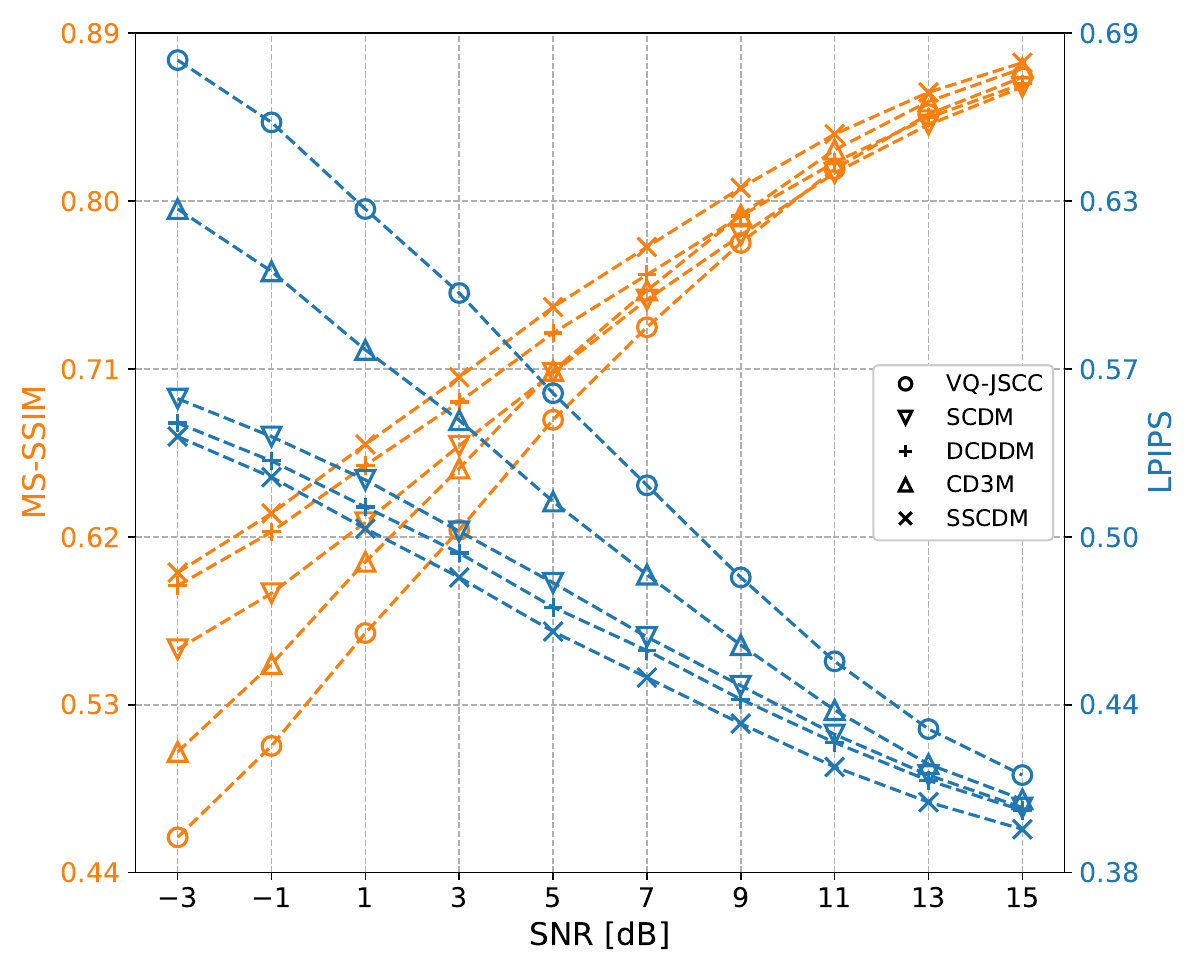}
    \label{fig:rayleigh}}
    \subfigure[64-QAM]{\includegraphics[width=0.32\linewidth]{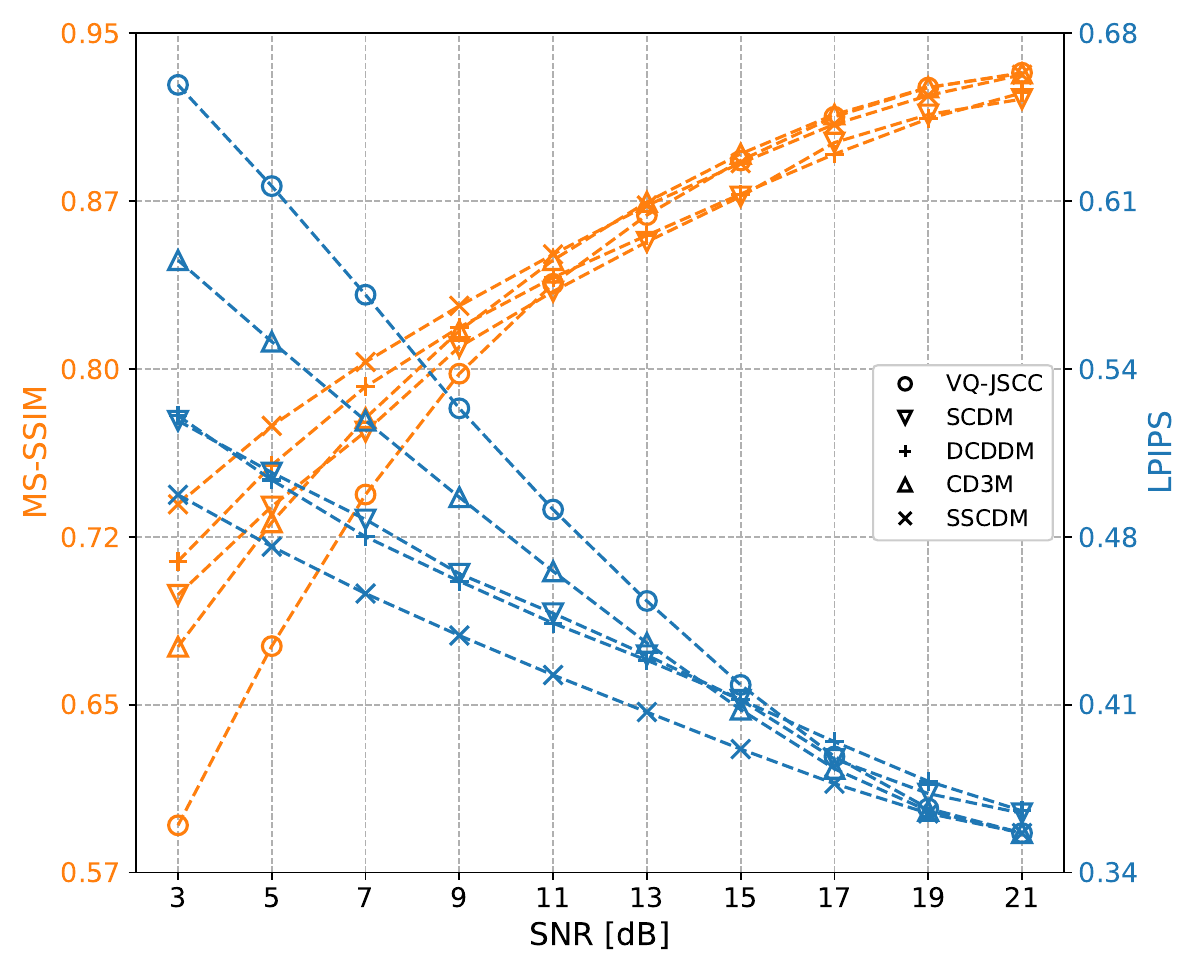}
    \label{fig:64-qam}}
    \subfigure[RDP]{\includegraphics[width=0.32\linewidth]{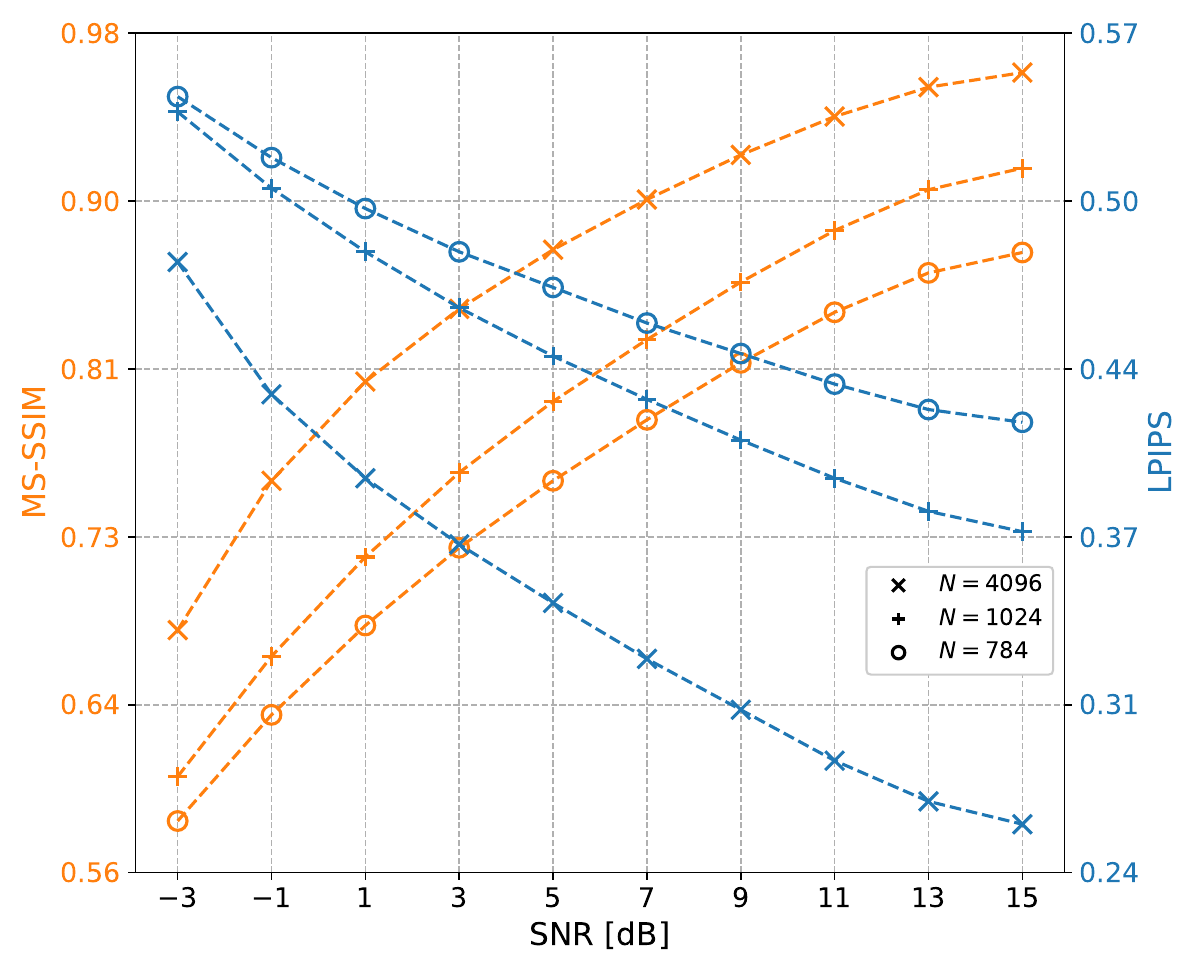}
    \label{fig:rdp}}
    \caption{MS-SSIM$\uparrow$ (orange) and LPIPS$\downarrow$ (blue) comparison between SSCDM and baseline methods (a) under Rayleigh fading channel, (b) using 64-QAM, and (c) varying compression ratios.}
    \label{fig:ablation_results}
\end{figure*}

Figs. \ref{fig:rayleigh} and \ref{fig:64-qam} evaluate the performance of SSCDM under Rayleigh fading and 64-QAM, respectively. For the Rayleigh fading channel, we assume that the receiver has perfect channel state information and applies matched filtering to equalize the received symbols before symbol detection. The x-axis represents the Tx SNR, $\mathbb{E}[|h|^2]P/\sigma^2$, where $h\sim\mathcal{CN}(0, 1)$ denotes the fading gain, while the diffusion starting point $k^*$ is selected according to the instantaneous Rx SNR, $|h|^2P/\sigma^2$, for each fading realization. For the 64-QAM setting, VQ-JSCC was trained at 25dB, and the DM-based correction model was trained to cover the dynamic SNR range from 3dB to 21dB. Additionally, we set $T' = 25$ for the subsequence $\{t_{k_\ell}\}_{\ell = 1}^{T'}$, with $k_\ell$ ranging from $4$ to $100$ at intervals of $4$, and change only $\lambda_2$ in \eqref{eq:bcd loss function} to $1$ while keeping the other hyperparameters unchanged. As shown in Figs. \ref{fig:rayleigh} and \ref{fig:64-qam}, SSCDM retains its performance advantage over the baseline methods under both settings. These results indicate that the proposed Markov-enforced symbol correction remains effective across different channel and modulation conditions. To further assess the robustness of SSCDM to channel estimation error, we model the estimated fading gain as $\hat{h} = h + n_h$, where $n_h\sim\mathcal{CN}(0,\sigma_h^2)$. $\hat{h}$ is used for both equalization and the selection of $k^*$. Table \ref{tab:channel estimation errors} reports the reconstruction performance for $\sigma_h\in\{0, 0.05, 0.1\}$. As $\sigma_h$ increases, the performance gradually degrades, while the degradation remains moderate.

\begin{table}[!t]
\centering
\caption{Performance of SSCDM under Rayleigh fading channel with different channel estimation errors.}
\label{tab:channel estimation errors}
\begin{tabular}{c|c|cccc}
\hline
\multirow{2}{*}{\textbf{Metric}} 
& \multirow{2}{*}{\textbf{$\sigma_{h}$}} 
& \multicolumn{4}{c}{SNR [dB]} \\
\cline{3-6}
& & $-3$ & $3$ & $9$ & $15$ \\
\hline
\multirow{3}{*}{\textbf{MS-SSIM}$\uparrow$}
& $0$ & 0.6008 & 0.7055 & 0.8070 & 0.8741 \\
& $0.05$ & 0.5995 & 0.7044 & 0.8044 & 0.8677 \\
& $0.1$ & 0.5938 & 0.6950 & 0.7860 & 0.8444 \\
\hline
\multirow{3}{*}{\textbf{LPIPS}$\downarrow$}
& $0$ & 0.5410 & 0.4890 & 0.4351 & 0.3964 \\
& $0.05$ & 0.5432 & 0.4918 & 0.4398 & 0.4046 \\
& $0.1$ & 0.5493 & 0.5036 & 0.4586 & 0.4260 \\
\hline
\end{tabular}
\end{table}

Fig. \ref{fig:rdp} exhibits the effect of varying the compression ratios in SSCDM by changing the number of transmitted feature vectors $N\in\{784, 1024, 4096\}$. Since $N$ corresponds to the number of transmitted symbols per image, the effective transmission rate is defined as $R_{\mathrm{eff}} = \frac{N}{CHW}$ channel uses per source dimension, consistent with the compression ratio adopted in our system. Across all SNR levels, using a smaller $N$ decreases MS-SSIM (a distortion-oriented metric) and increases LPIPS (a perception-oriented metric) due to increased information loss. This rate-dependent behavior provides a rate-distortion-perception (RDP)-oriented interpretation, where reducing the effective transmission rate can degrade both distortion- and perception-oriented performance, though not a full characterization of the RDP tradeoff. This observation mainly comes from the considered discrete DM-based semantic communication system for image reconstruction, rather than from the proposed Markov-enforced transition design itself. Therefore, the integration of SSCDM into the considered image-transmitting semantic communication system preserves the expected RDP behavior, while the core objective of the proposed Markov-enforced transition design is to improve digital symbol correction performance.

\subsection{Computational Complexity}\label{subsec:computational complexity}

\begin{table}[!t]
\centering
\caption{Inference FLOPs, latency, and effective NFEs.}
\label{tab:inference complexity}
\begin{tabular}{c|c|ccc}
\hline
\multirow{2}{*}{\textbf{Metric}} 
& \multirow{2}{*}{\textbf{Method}} 
& \multicolumn{3}{c}{SNR [dB]} \\
\cline{3-5}
& & $-3$ & $3$ & $9$ \\
\hline
\multirow{4}{*}{\textbf{GFLOPs}}
& SCDM  & 519.35  & 448.74  & 378.13 \\
& DCDDM & 360.13  & 280.84  & 222.94 \\
& CD3M  & 3102.83 & 1786.31 & 917.77 \\
& SSCDM & 360.44  & 311.62  & 116.35 \\
\hline
\multirow{4}{*}{\textbf{Latency [ms]}}
& SCDM  & 119.01 & 100.76 & 83.56 \\
& DCDDM & 71.08  & 51.254  & 38.13 \\
& CD3M  & 610.90 & 402.08 & 172.96 \\
& SSCDM & 69.87  & 54.18  & 12.52 \\
\hline
\multirow{4}{*}{\textbf{Effective NFEs}}
& SCDM  & 138  & 117 & 96 \\
& DCDDM & 100  & 74  & 55 \\
& CD3M  & 1000 & 568 & 283 \\
& SSCDM & 100  & 84  & 20 \\
\hline
\end{tabular}
\end{table}

\begin{table}[!t]
\centering
\caption{Training FLOPs and time.}
\label{tab:training complexity}
\begin{tabular}{c|cccc}
\hline
\textbf{Metric} 
& SCDM 
& DCDDM 
& CD3M 
& SSCDM \\
\hline
\textbf{GFLOPs}
& 3.3625 & 3.0475 & 3.0475 & 3.0511 \\
\textbf{Time [ms]}
& 136.61 & 80.49 & 81.10 & 77.28 \\
\hline
\end{tabular}
\end{table}

Table \ref{tab:inference complexity} reports the inference FLOPs, inference latency, and effective number of function evaluations (NFEs) per image during DM-based symbol correction. For a fair comparison, all methods are evaluated on an NVIDIA RTX 3090 GPU, and each benchmark method uses the sampling-step specification reported in its original paper for best reconstruction performance. CD3M exhibits high inference complexity due to its large number of diffusion steps. SCDM also incurs high inference latency because of the predictor-corrector sampling procedure of the continuous DM. DCDDM requires fewer NFEs than CD3M and SCDM, but its linear scheduling still leads to relatively large NFEs in the high-SNR regime. In contrast, SSCDM achieves comparable or lower inference complexity than the baselines while preserving the superior reconstruction performance shown in Fig. \ref{fig:comparison_results}.

Table \ref{tab:training complexity} further compares the training FLOPs and time per image, corresponding to a single diffusion step. SCDM requires a longer training time because it employs the more complex NCSN++ architecture \cite{song2020score}. Since DCDDM, CD3M, and SSCDM use the same discrete DM architecture, they exhibit similar training complexity. This confirms that the performance gain of SSCDM mainly comes from the proposed Markov-enforced transition design rather than from additional training complexity.

\begin{table}[!t]
\centering
\caption{Performance of SSCDM using different maximum time steps $T$.}
\label{tab:diffusion time steps}
\begin{tabular}{c|c|cccc}
\hline
\multirow{2}{*}{\textbf{Metric}} 
& \multirow{2}{*}{\textbf{$T$}} 
& \multicolumn{4}{c}{SNR [dB]} \\
\cline{3-6}
& & $-3$ & $3$ & $9$ & $15$ \\
\hline
\multirow{3}{*}{\textbf{MS-SSIM}$\uparrow$}
& $100$ & 0.6080 & 0.7601 & 0.8554 & 0.9124 \\
& $200$ & 0.6010 & 0.7596 & 0.8553 & 0.9117 \\
& $500$ & 0.5903 & 0.7618 & 0.8568 & 0.9113 \\
\hline
\multirow{3}{*}{\textbf{LPIPS}$\downarrow$}
& $100$ & 0.5389 & 0.4619 & 0.4102 & 0.3742 \\
& $200$ & 0.5432 & 0.4627 & 0.4104 & 0.3739 \\
& $500$ & 0.5510 & 0.4627 & 0.4101 & 0.3741 \\
\hline
\end{tabular}
\end{table}

We also investigate the effect of the maximum time step $T$, as shown in Table \ref{tab:diffusion time steps}. Increasing $T$ from 100 to 200 or 500 does not provide consistent gains in either MS-SSIM or LPIPS, indicating that $T = 100$ is sufficient for effective symbol correction in most SNR regimes. Therefore, we set $T = 100$ in all experiments as a balanced choice between reconstruction performance and inference latency.

\subsection{Impact of Codebook-Guided U-Net Embedding}\label{subsec:impact of codebook-guided u-net embedding}

\begin{figure}[!t]
    \centering
    \includegraphics[width=0.8\linewidth]{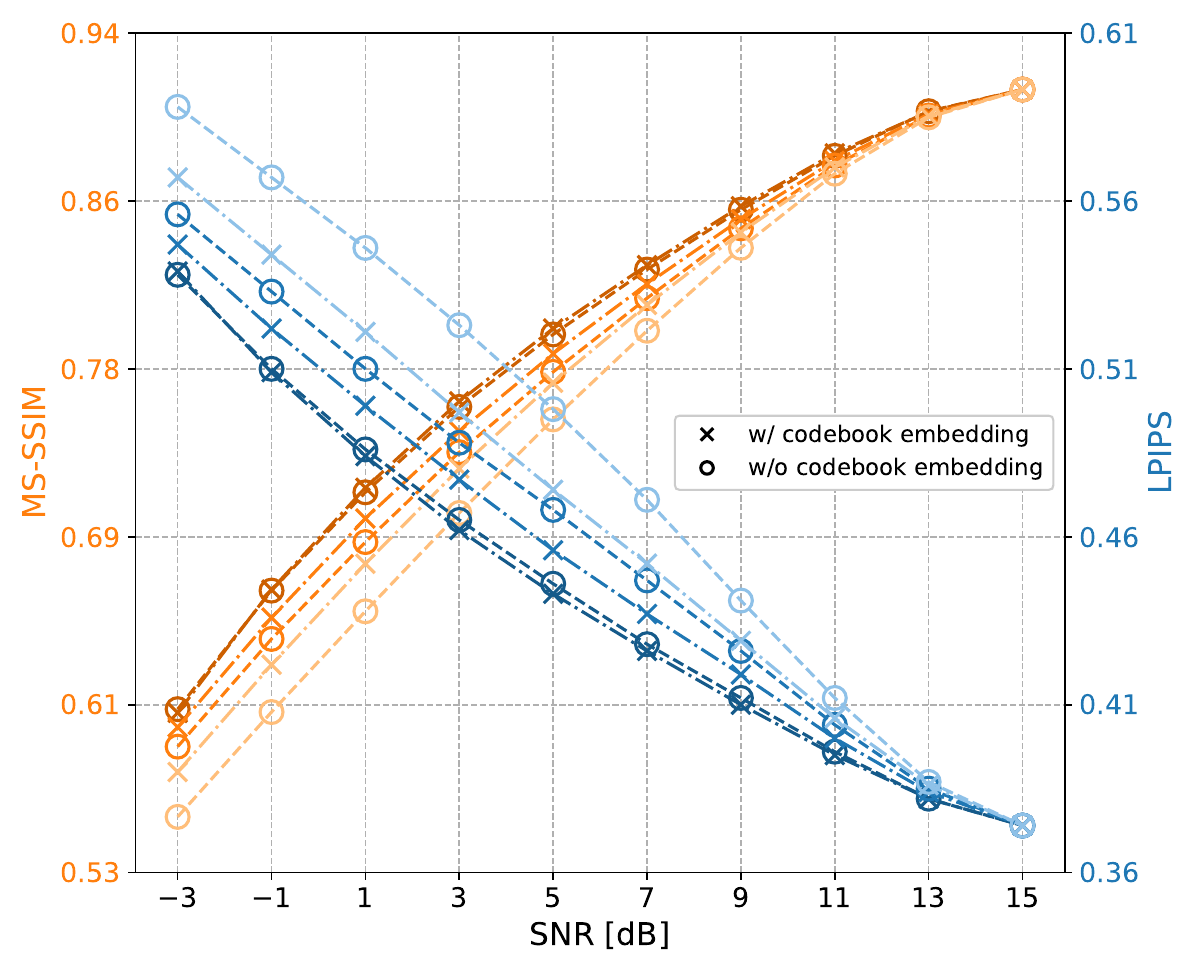}
    \caption{MS-SSIM$\uparrow$ (orange) and LPIPS$\downarrow$ (blue) comparison between SSCDM with and without codebook embedding across different base tensor channel dimensions $\kappa$ of the U-Net. Darker colors are larger $\kappa$ (16 $\rightarrow$ 32 $\rightarrow$ 64).}
    \label{fig:ffhq_codeword_semantics}
\end{figure}

As described in Section \ref{subsec:semantics-informed diffusion model architecture}, we incorporate the codebook embedding into the U-Net input to strengthen the integration of the DM within the overall JSCC framework. By aligning the diffusion process with the latent semantic space of the JSCC decoder, this design enables the DM to more effectively capture semantic relationships among codewords. Fig. \ref{fig:ffhq_codeword_semantics} clearly shows that the embedding-based approach consistently outperforms the one without embedding. The performance gap becomes more pronounced as $\kappa$ decreases, i.e., the model size is reduced, indicating a limitation in representational capacity. By injecting codebook-level semantic structure, the proposed embedding enables the DM to exploit richer semantic information and mitigates performance degradation in compact models. As a result, the DM achieves more stable and semantically coherent symbol correction with only a negligible parameter overhead of $1.8$K.

\subsection{Effect of SOM-Based Codebook Learning}\label{subsec:effect of som-based codebook learning}

\begin{figure}[!t]
    \centering
    \includegraphics[width=0.9\linewidth]{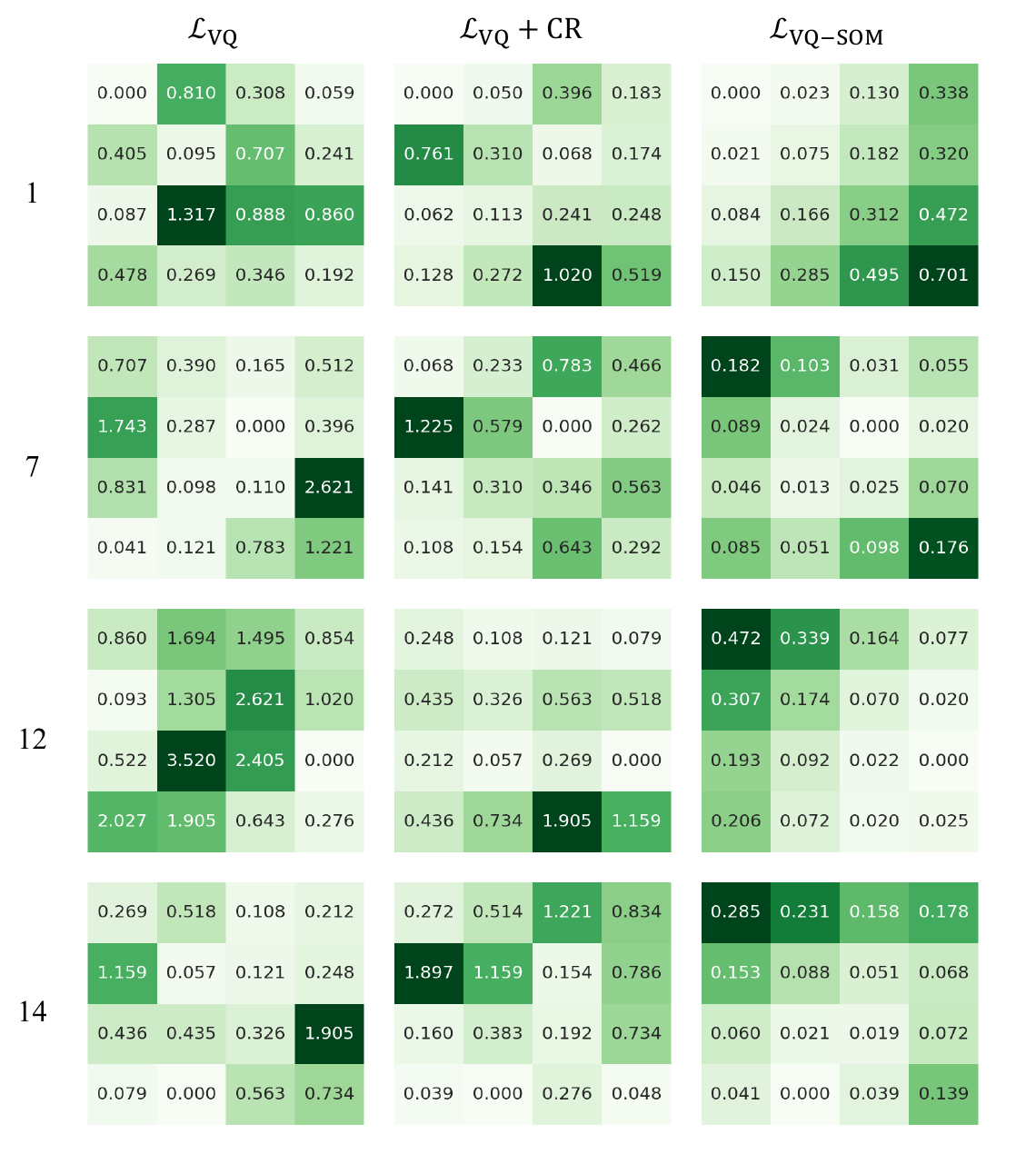}
    \caption{Heatmaps of Euclidean distances between each codeword constructed by using $\mathcal{L}_\mathsf{VQ}$, $\mathcal{L}_\mathsf{VQ}$ + CR, and $\mathcal{L}_\mathsf{VQ\text{-}SOM}$.}
    \label{fig:codeword_similarity}
\end{figure}

Fig. \ref{fig:codeword_similarity} illustrates the Euclidean distances between the learned codewords and four randomly selected reference codewords, i.e., $\mathbf{c}_1$, $\mathbf{c}_7$, $\mathbf{c}_{12}$, and $\mathbf{c}_{14}$. When trained solely with $\mathcal{L}_\mathsf{VQ}$, no mechanism enforces local consistency in the feature space. Consequently, the distance pattern does not exhibit a smoothly increasing transition from the reference codewords. The codebook reordering (CR) algorithm \cite{zhou2024moc} reorders codewords by greedily finding the minimum distance codeword and then applies Gray coding after the codebook optimization with $\mathcal{L}_\mathsf{VQ}$. Although the CR approach attempts to enhance the locality, it still fails to preserve meaningful spatial relationships between codewords. This is mainly because this method is heuristic and lacks a learning-based adaptation. In contrast, our SOM-based method trained with $\mathcal{L}_{\mathsf{VQ\text{-}SOM}}$ successfully establishes a smooth and topology-preserving structure in the codebook. The heatmaps clearly demonstrate this property through gradual color transitions, indicating progressively increasing distances from the reference codewords.

\begin{figure}[!t]
    \centering
    \includegraphics[width=0.8\linewidth]{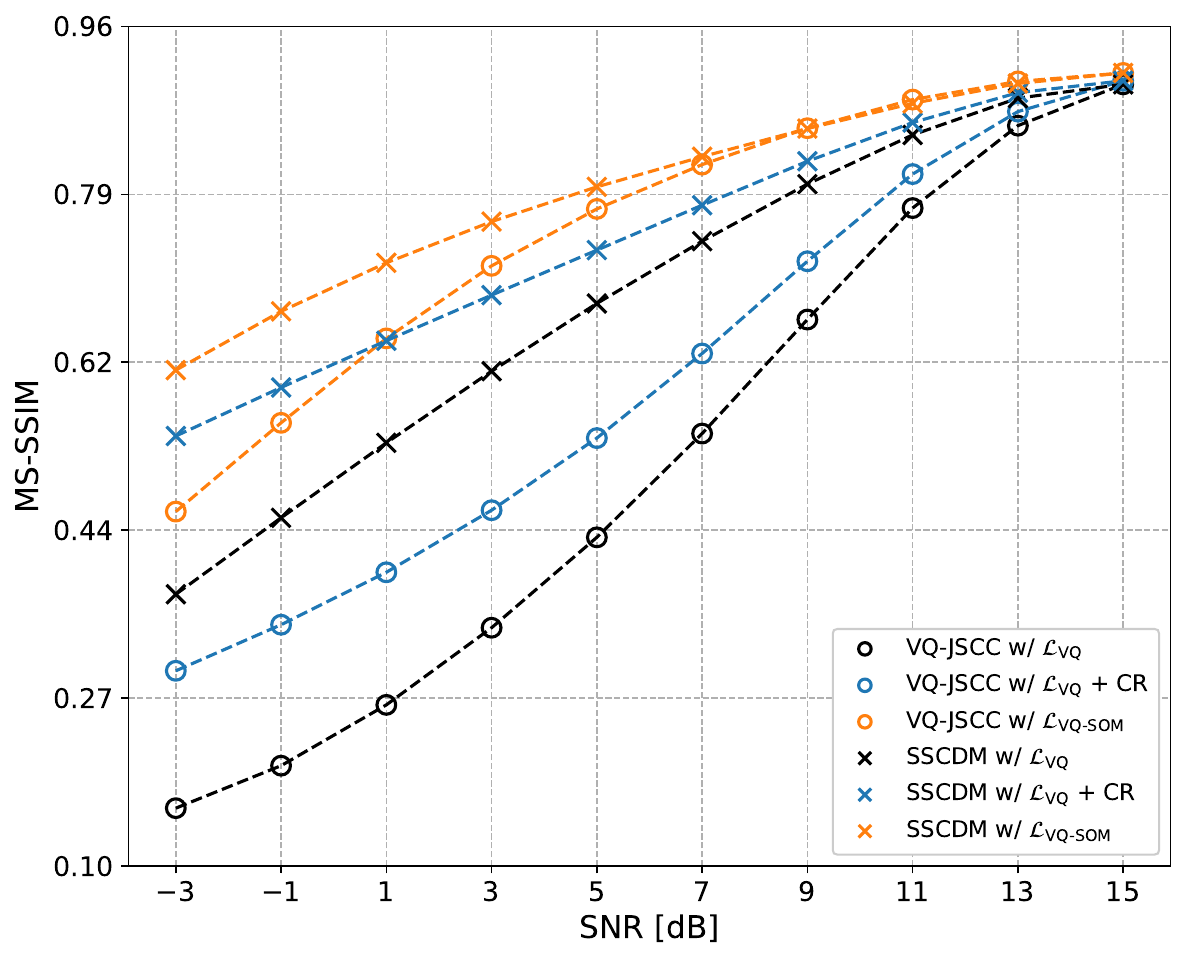}
    \caption{MS-SSIM$\uparrow$ comparison between VQ codebook construction methods without and with symbol error correction by SSCDM.}
    \label{fig:ffhq_som_cr_naive}
\end{figure}

The Euclidean distance property illustrated in Fig. \ref{fig:codeword_similarity} directly translates into the image reconstruction performance shown in Fig. \ref{fig:ffhq_som_cr_naive}. Specifically, the SOM-aided SSCDM achieves remarkably better performance than SSCDM with $\mathcal{L}_\mathsf{VQ}$ or combined with the CR algorithm. In the two baseline variants using SSCDM with $\mathcal{L}_\mathsf{VQ}$, the VQ codebook fails to maintain local consistency in the constellation space, preventing SSCDM from effectively linking its error correction capability to downstream task performance. In contrast, our SOM-based learning produces a topology-preserving codebook, enabling SSCDM’s error correction process to more effectively enhance image reconstruction performance compared to VQ-JSCC with $\mathcal{L}_{\mathsf{VQ\text{-}SOM}}$. Moreover, performance improvements are consistently observed across all codebook construction methods when integrated with SSCDM, confirming the effective error-correcting capability of the proposed SSCDM.

\section{Conclusion}\label{sec:conclusion}
This work introduced SSCDM, a novel semantic symbol correcting DM tailored for VQ-based digital semantic communication under a discrete-time diffusion framework. By enforcing Markov-consistent transition dynamics for discrete symbols using a shared-eigenbasis CTMC solution structure and embedding them into a semantic latent space, our approach transforms the inherently non-Markovian symbol transitions induced by realistic digital communication channels into a tractable form, enabling principled and robust symbol correction. Moreover, a SOM-based regularization is incorporated into the codebook learning to preserve the geometric vicinity of neighboring symbols, further enhancing correction reliability and reconstruction fidelity. Comprehensive experiments under various channel conditions and datasets confirm that SSCDM outperforms existing baselines, particularly in low-SNR regimes. We believe this work advances the integration of discrete DMs into digital semantic communication. Future research could extend the correction framework to point-to-point multi-modal settings, where discrete DMs jointly process text, images, videos, and cross-modal representations for symbol error correction. Another direction is multiple access scenarios, where symbol errors are jointly corrected across users under interference channels, further broadening the applicability of discrete DM-based correction.

\appendices
\section{Proof of Theorem \ref{thm:Markov is impossible}}\label{apdx:Markov is impossible}
We prove the theorem by comparing the two expressions $[\mathbf{Q}_{t_k|t_0}]_{i, j}$ and $[\mathbf{Q}_{t_{k - 1}|t_0}\mathbf{Q}_{t_k|t_{k - 1}}]_{i, j}$, both of which describe the transition probability from a single symbol point $\mathbf{s}_i$ at time point $t_0$ to region $\mathbb{A}_j$ at time point $t_k$. Showing that these two probabilities are generally unequal is equivalent to showing that the Chapman-Kolmogorov equation is violated.

By the definition of cumulative variance $\bar{\sigma}_{t_k}^2 = \bar{\sigma}_{t_{k - 1}}^2 + \sigma_{t_k}^2$, the direct transition $\mathbf{Q}_{t_k|t_0}$ is given by
\begin{align}\label{eq:ckleft}
    \nonumber[\mathbf{Q}_{t_k|t_0}]_{i, j} &= \iint_{\mathbb{A}_j}\frac{1}{\pi\bar{\sigma}_{t_k}^2}e^{-\frac{||\mathbf{x} - \mathbf{s}_i||_2^2}{\bar{\sigma}_{t_k}^2}}d\mathbf{x}\\
    &=\sum_{u = 1}^M\iint_{\mathbb{A}_u}\frac{1}{\pi\bar{\sigma}_{t_{k - 1}}^2}e^{-\frac{||\mathbf{x}' - \mathbf{s}_i||_2^2}{\bar{\sigma}_{t_{k - 1}}^2}}\\
    &\nonumber\quad\quad\quad\quad\times\iint_{\mathbb{A}_j}\frac{1}{\pi\sigma_{t_k}^2}e^{-\frac{||\mathbf{x} - \mathbf{x}'||_2^2}{\sigma_{t_k}^2}}d\mathbf{x}d\mathbf{x}'.
\end{align}
The second equality reflects a sequential corruption process, where two independent Gaussian noises $n_1 \sim \mathcal{CN}(0, \bar{\sigma}_{t_{k - 1}}^2)$ and $n_2 \sim \mathcal{CN}(0, \sigma_k^2)$ are successively applied to the symbol $\mathbf{s}_i$. Each term in the summation over $u\in[1{:}M]$ represents a path where the first noise perturbs $\mathbf{s}_i$ to a specific point $\mathbf{x}'\in\mathbb{A}_u$, and the second noise moves this exactly same point $\mathbf{x}'$ into the region $\mathbb{A}_j$. In other words, the product of probabilities is taken for the same intermediate point $\mathbf{x}'$ before integration.

On the other hand, the matrix product yields
\begin{align}\label{eq:ckright}
    &\nonumber[\mathbf{Q}_{t_{k - 1}|t_0}\mathbf{Q}_{t_k|t_{k - 1}}]_{i, j} = \sum_{u = 1}^M\iint_{\mathbb{A}_u}\frac{1}{\pi\bar{\sigma}_{t_{k - 1}}^2}e^{-\frac{||\mathbf{x} - \mathbf{s}_i||_2^2}{\bar{\sigma}_{t_{k - 1}}^2}}d\mathbf{x}\\
    &\quad\times\iint_{\mathbb{A}_u}\frac{a_{t_{k - 1}}(\mathbf{x}'')}{b_{t_{k - 1}}(u)}\iint_{\mathbb{A}_j}\frac{1}{\pi\sigma_{t_k}^2}e^{-\frac{||\mathbf{x}' - \mathbf{x}''||_2^2}{\sigma_{t_k}^2}}d\mathbf{x}'d\mathbf{x}''.
\end{align}
In this case, the integral over the intermediate region $\mathbb{A}_u$ effectively includes contributions from products of probabilities corresponding to different points $\mathbf{x}$ and $\mathbf{x}''$ within the same region. That is, transitions from $\mathbf{s}_i$ at $t_0$ to $\mathbf{x}\in\mathbb{A}_u$ at $t_{k - 1}$ combined with $\mathbf{x}''\in\mathbb{A}_u$ at $t_{k - 1}$ to $\mathbf{x}'\in\mathbb{A}_j$ at $t_k$ for $\mathbf{x} \neq \mathbf{x}''$ also enter the calculation.

Consequently, the two expressions are mathematically and fundamentally distinct, since the difference in their underlying $\sigma$-algebras directly causes them to integrate over \textit{different sets of events}. The true probability in \eqref{eq:ckleft} is conditioned on the fine-grained, point-refined $\sigma$-algebra $\mathcal{F}_{t_{k-1}} = \sigma(X_{t_{k-1}})$. Conditioning on $\mathcal{F}$ necessitates computing \textit{an integral of a product}, thereby correctly accounting only for paths that pass through the \textit{same intermediate point} $\mathbf{x}'$, conceptually expressed as
\begin{align}
    \nonumber(\mathbf{s}_i\to\mathbb{A}_u\to\mathbb{A}_j) &= \sum\nolimits_{\mathbf{x}'\in\mathbb{A}_u}(\mathbf{s}_i\to\mathbf{x}'\to\mathbb{A}_j)\\
    &= \sum\nolimits_{\mathbf{x}'\in\mathbb{A}_u}(\mathbf{s}_i\to\mathbf{x}')\times(\mathbf{x}'\to\mathbb{A}_j).
\end{align}
In contrast, the Chapman-Kolmogorov composition in \eqref{eq:ckright} forces conditioning on the coarse, region-based $\sigma$-algebra $\mathcal{G}_{t_{k-1}} = \sigma(\{X_{t_{k-1}}\in\mathbb{A}_u\})$. This conditioning on $\mathcal{G}$, i.e., a strictly coarser sub-$\sigma$-algebra of $\mathcal{F}$, compels the calculation into the form of \textit{a product of integrals}. This structure inherently includes not only paths through the \textit{same intermediate point}, but also mixed paths having \textit{different arrival and departure points} $\mathbf{x} \neq \mathbf{x}''$ at $t_{k - 1}$, informally written as
\begin{align}
    \nonumber(\mathbf{s}_i\to\mathbb{A}_u\to\mathbb{A}_j) &= \sum_{\mathbf{x}\in\mathbb{A}_u}(\mathbf{s}_i\to\mathbf{x})\times\sum_{\mathbf{x}''\in\mathbb{A}_u}(\mathbf{x}''\to\mathbb{A}_j)\\
    &= \sum_{\mathbf{x}, \mathbf{x}''\in\mathbb{A}_u}(\mathbf{s}_i\to\mathbf{x})\times(\mathbf{x}''\to\mathbb{A}_j).
\end{align}
Intuitively, this discrepancy arises because the point-level information at $t_{k - 1}$, which is required to determine the next transition and preserve the Markov property, is no longer available once the intermediate state is represented only by the region-level information $\{X_{t_{k-1}}\in\mathbb{A}_u\}$. The same reasoning applies to all intermediate regions $\{\mathbb{A}_u\}_{u = 1}^M$.

The two probabilities in \eqref{eq:ckleft} and \eqref{eq:ckright} are defined over \textit{non-identical event descriptions}, and the relevant probability integrands are non-constant over decision regions with nonzero measure for typical digital modulation schemes under nonzero Gaussian noise. Hence, the point-conditioned composition in \eqref{eq:ckleft} does not coincide with the region-conditioned composition in \eqref{eq:ckright}, i.e., $[\mathbf{Q}_{t_k|t_0}]_{i, j} \neq [\mathbf{Q}_{t_{k - 1}|t_0}\mathbf{Q}_{t_k|t_{k - 1}}]_{i, j}$.\footnote{Specially constructed nonstandard decision regions or constant transition probability functions within each region may lead to equality. However, such cases are atypical and fall outside the digital communication setting considered in this work.} This mismatch is also numerically verified by MC simulations in \textbf{Remark \ref{rm:Markov is impossible}}. Therefore, the Chapman-Kolmogorov equation is violated, establishing the non-Markovian nature of the detection-induced symbol transition process.

\section{Proof of Lemma \ref{lm:non-increasing diagonal}}\label{apdx:non-increasing diagonal}
To investigate the sign of the eigenvalues of the rate matrix, we introduce a matrix $\mathbf{P} = \frac{1}{\rho}\overrightarrow{\mathbf{R}_t} + \mathbf{I}_M\in\mathbb{R}^{M\times M}$, where $\rho$ is a positive scalar satisfying $\rho > \max_i |[\overrightarrow{\mathbf{R}_t}]_{i, i}|$. Under this choice, $\mathbf{P}$ is row-stochastic. Since $\frac{1}{\rho}\overrightarrow{\mathbf{R}_{t}}\in\mathbb{R}^{M\times M}$ has non-positive diagonal entries strictly greater than $-1$ and non-negative off-diagonal entries according to \textbf{Proposition \ref{prop:rate matrix properties}}, all entries of $\mathbf{P}$ are non-negative. Also, as $\overrightarrow{\mathbf{R}_t}$ has the zero row-sum property from \textbf{Proposition \ref{prop:rate matrix properties}}, i.e., $\sum_{j = 1}^M [\overrightarrow{\mathbf{R}_t}]_{i, j} = 0$ for $i\in[1{:}M]$, $\mathbf{P}$ is row-stochastic with satisfying $\sum_{j = 1}^M [\mathbf{P}]_{i, j} = 1$ for $i\in[1{:}M]$. Then, since the rate matrix can be eigen-decomposed into $\overrightarrow{\mathbf{R}_t} = \mathbf{V}\mathbf{\Sigma}(t)\mathbf{V}^{-1}$ as in \eqref{eq:ctmc rate matrix}, we can express $\mathbf{P}$ as
\begin{align}
    \mathbf{P} = \frac{1}{\rho}\overrightarrow{\mathbf{R}_{t}} + \mathbf{I}_M = \mathbf{V}(\frac{1}{\rho}\mathbf{\Sigma}(t) + \mathbf{I}_M)\mathbf{V}^{-1} = \mathbf{V}\mathbf{\Lambda}\mathbf{V}^{-1},
\end{align}
where $\mathbf{\Lambda} = \frac{1}{\rho}\mathbf{\Sigma}(t) + \mathbf{I}_M\in\mathbb{R}^{M\times M}$.

\begin{lemma}\label{lm:maximum eigenvalue}
    The largest eigenvalue of any row-stochastic matrix is $1$.
\begin{IEEEproof}
    See \emph{Appendix \ref{apdx:maximum eigenvalue}}.
\end{IEEEproof}
\end{lemma}

Since $\mathbf{P}$ is row-stochastic, \textbf{Lemma \ref{lm:maximum eigenvalue}} implies that $[\mathbf{\Lambda}]_{i, i} \leq 1$ for $i\in[1{:}M]$. Therefore, because $\mathbf{\Lambda} = \frac{1}{\rho}\mathbf{\Sigma}(t) + \mathbf{I}_M$, each diagonal element of $\mathbf{\Sigma}(t)$ is non-positive. Recall that $\mathbf{D}(t) = \mathsf{exp}(\int_0^t\mathbf{\Sigma}(t')dt')$, where $\mathbf{D}(t)$ comes from the eigen-decomposition of $\bar{\mathbf{Q}}_{t|0} = \mathbf{V}\mathbf{D}(t)\mathbf{V}^{-1}$ as in \eqref{eq:ctmc transition matrix}. As a result, the integral $\int_0^t \mathbf{\Sigma}(t')dt'$ is non-increasing in $t$, and thus each diagonal element $[\mathbf{D}(t)]_{i, i}$ is also non-increasing in $t$. Hence, the sequence $\{[\mathbf{D}(t_k)]_{i,i}\}_{k=0}^T$ is non-increasing for $i\in[1{:}M]$.

\section{Proof of Lemma \ref{lm:maximum eigenvalue}}\label{apdx:maximum eigenvalue}
Let $\mathbf{P}\in\mathbb{R}^{M\times M}$ be a row-stochastic matrix, i.e., $[\mathbf{P}]_{j, :}\mathbf{1}_M = 1$ for $j\in[1{:}M]$, and suppose it admits the eigen-decomposition $\mathbf{V}\mathbf{\Lambda}\mathbf{V}^{-1}$. For contradiction, suppose that there exists an eigenvalue $[\mathbf{\Lambda}]_{i, i} > 1$ for some $i\in[1{:}M]$, where $[\mathbf{V}]_{:, i}$ is the corresponding eigenvector. Without loss of generality, let $j$ be an index such that $[\mathbf{V}]_{j, i} = \max_{k\in[1{:}M]} [\mathbf{V}]_{k, i} > 0$.
Then, by the eigenvalue relation, we have $[\mathbf{P}]_{j, :}[\mathbf{V}]_{:, i} = [\mathbf{V}]_{j, i}[\mathbf{\Lambda}]_{i, i} > [\mathbf{V}]_{j, i}$, which contradicts $[\mathbf{P}]_{j, :}[\mathbf{V}]_{:, i} \leq [\mathbf{P}]_{j, :}([\mathbf{V}]_{j, i}\mathbf{1}_M) = [\mathbf{V}]_{j, i}([\mathbf{P}]_{j, :}\mathbf{1}_M) = [\mathbf{V}]_{j, i}$. Therefore, no eigenvalue of $\mathbf{P}$ can exceed $1$, and since $\mathbf{P} \mathbf{1}_M = 1\cdot\mathbf{1}_M$, the largest eigenvalue is exactly $1$.

\section{Noise Scheduling Method}\label{apdx:noise scheduling method}
Our noise scheduling follows the sigmoid-based formula presented in Algorithm 4 of \cite{jabri2022scalable}. The standard form of the scheduling function for $0 \leq t \leq 1$ is expressed as
\begin{align}\label{eq:sigmoid scheduling}
    \omega_t = \frac{\mathsf{sig}(t(\nu_{\mathsf{end}} - \nu_{\mathsf{start}}) + \nu_{\mathsf{start}}) - \mathsf{sig}(\nu_{\mathsf{start}})}{\mathsf{sig}(\nu_{\mathsf{end}}) - \mathsf{sig}(\nu_{\mathsf{start}})},
\end{align}
where $\mathsf{sig}(x) = {1}/({1 + e^{-x}})$ denotes the sigmoid function, and $\nu_{\mathsf{start}}$ and $\nu_{\mathsf{end}}$ control the curvature of the scheduling curve. We employ this formulation to design a noise scheduling function that outputs the SNR $\bar{\eta}_{t_k}$ in dB scale at time points $t_k$ with $k\in[1{:}T]$, where $t_0 = 0$ and $0 < t_k \leq 1$. Accordingly, for $k\in[1{:}T]$, $\bar{\eta}_{t_k}$ in dB scale is defined as
\begin{align}
   \bar{\eta}_{t_k} [\text{dB}] = \xi_1\times10\cdot\mathsf{log}_{10}\!\left(\frac{1 - \omega_{t_{k - 1}}}{\omega_{t_{k - 1}}}\right) + \xi_2,
\end{align}
where $\xi_1$ and $\xi_2$ are scaling and offset parameters, respectively. The relationship between $\bar{\eta}_{t_k} [\text{dB}]$ and the noise variance sequence $\{\sigma_{t_k}\}_{k = 1}^T$ is given by $\bar{\eta}_{t_k} [\text{dB}] = 10\cdot\mathsf{log}_{10}(P/\bar{\sigma}_{t_k}^2)$, where $\bar{\sigma}_{t_k}^2 = \sum\nolimits_{k' = 1}^k\sigma_{t_{k'}}^2$. In our experiments, we set $\nu_{\mathsf{start}} = 0.025$, $\nu_{\mathsf{end}} = 1.25$, $\xi_1 = 0.45$ and $\xi_2 = 6.5$ for 16-QAM, and $\nu_{\mathsf{start}} = 0.0125$, $\nu_{\mathsf{end}} = 0.625$, $\xi_1 = 0.7$ and $\xi_2 = 11$ for 64-QAM.

%

\bibliographystyle{IEEEtran}  
\bibliography{IEEEabrv,reference}

\end{document}